\documentclass[10pt,journal,compsoc]{IEEEtran}
\IEEEoverridecommandlockouts
\usepackage[linesnumbered, vlined, ruled]{algorithm2e}
\usepackage{algorithmicx}
\usepackage[noend]{algpseudocode}
\usepackage{graphicx}
\usepackage{amsmath,amssymb,amsfonts}
\usepackage{amsthm}
\usepackage{booktabs}
\usepackage{hhline}
\usepackage{xcolor}
\usepackage{autonum}
\usepackage{subfigure}

\newtheorem{theorem}{Theorem}[section]

\newtheorem{definition}{Definition}[section]

\definecolor{BrickRed}{HTML}{b84a39} 
\definecolor{SpecialBlue}{HTML}{004e66}

\makeatletter
\def\ps@IEEEtitlepagestyle{%
  \def\@oddfoot{\mycopyrightnotice}%
  \def\@evenfoot{}%
}
\def\mycopyrightnotice{%
  {\footnotesize 
  \begin{minipage}{\textwidth}
  \centering
  \textbf{This work has been submitted to the IEEE for possible publication. Copyright may be transferred without notice, after which this version may no longer be accessible.}
  \end{minipage}
    \hfill}
  \gdef\mycopyrightnotice{}
}

\usepackage[
	n,
	operators,
	advantage,
	sets,
	adversary,
	landau,
	probability,
	notions,	
	ff,
	mm,
	primitives,
	events,
	complexity,
	asymptotics,
	keys]{cryptocode}

\begin{document}

\title{CloudChain: A Cloud Blockchain Using Shared Memory Consensus and RDMA}

\author{Minghui~Xu,~\IEEEmembership{Member,~IEEE,}
        Shuo~Liu,~\IEEEmembership{Student Member,~IEEE,}
        Dongxiao~Yu,~\IEEEmembership{Senior Member,~IEEE,}
        Xiuzhen~Cheng,~\IEEEmembership{Fellow,~IEEE,}
        Shaoyong~Guo,~
        and~Jiguo~Yu,~\IEEEmembership{Senior Member,~IEEE}
\thanks{M. Xu, S. Liu, D. Yu (Corresponding Author),  and X.  Cheng are with the School of Computer Science and Technology, Shandong University, Qingdao, 266510, P. R. China. E-mail: mhxu@sdu.edu.cn; 201800130061@mail.sdu.edu.cn; dxyu@sdu.edu.cn,  xzcheng@sdu.edu.cn. 
}
\thanks{S. Guo is with the State Key Laboratory of Networking and Switching Technology,  Beijing University of Posts and Communications, Beijing 100876, China.  Email: syguo@bupt.edu.cn.
}
\thanks{J. Yu is with the Qilu University of Technology (Shandong Academy of
Sciences), Jinan, Shandong, 250353, P. R. China; the Shandong Computer
Science Center (National Supercomputer Center in Jinan), Jinan, Shandong,
250014, P. R. China; and the Shandong Laboratory of Computer Networks,
Jinan, 250014, P. R. China. Email: jiguoyu@sina.com.}
\thanks{Manuscript received April 19, 2005; revised August 26, 2015.}}

\markboth{Journal of \LaTeX\ Class Files,~Vol.~14, No.~8, August~2015}%
{Shell \MakeLowercase{\textit{et al.}}: Bare Demo of IEEEtran.cls for IEEE Journals}

\IEEEtitleabstractindextext{
\begin{abstract}
	Blockchain technologies can enable secure computing environments among mistrusting parties.  Permissioned blockchains are particularly enlightened by companies, enterprises, and government agencies due to their efficiency, customizability, and governance-friendly features.  Obviously, seamlessly fusing blockchain and cloud computing can significantly benefit permissioned blockchains; nevertheless, most blockchains implemented on clouds are originally designed for loosely-coupled networks where nodes communicate asynchronously, failing to take advantages of the closely-coupled nature of cloud servers.  In this paper, we propose an innovative cloud-oriented blockchain -- CloudChain, which is a modularized three-layer system composed of the network layer, consensus layer, and blockchain layer. CloudChain is based on a shared-memory model where nodes communicate synchronously by direct memory accesses.  We realize the shared-memory model with the Remote Direct Memory Access  technology, based on which we propose a shared-memory consensus algorithm to ensure presistence and liveness, the two crucial blockchain security properties countering Byzantine nodes.  We also implement a CloudChain prototype based on a RoCEv2-based testbed to experimentally validate our design, and the results verify the feasibility and efficiency of CloudChain.   
\end{abstract}

}

\maketitle

\IEEEpeerreviewmaketitle

\section{Introduction}
\label{sec:intro}

Blockchain can create a trusted computing environment among mutually mistrusting parties, providing them with remarkable properties such as decentralization, immutability, and traceability.  As a result,  blockchain has delivered benefits to many uses beyond cryptocurrency.  
Based on openness, blockchains can be categorized as either permissionless or permissioned.  Permissionless blockchains allow users to freely join or leave the network while permissioned ones require permissions from the blockchain owner to participate.  Therefore, permissioned blockchains are favored by companies,  enterprises, and government agencies who are willing to build efficient, secure, cost-effective blockchains but reluctant to share data with the public. 

\begin{figure}[!htbp]
	\centering
	\includegraphics[width=3.5in]{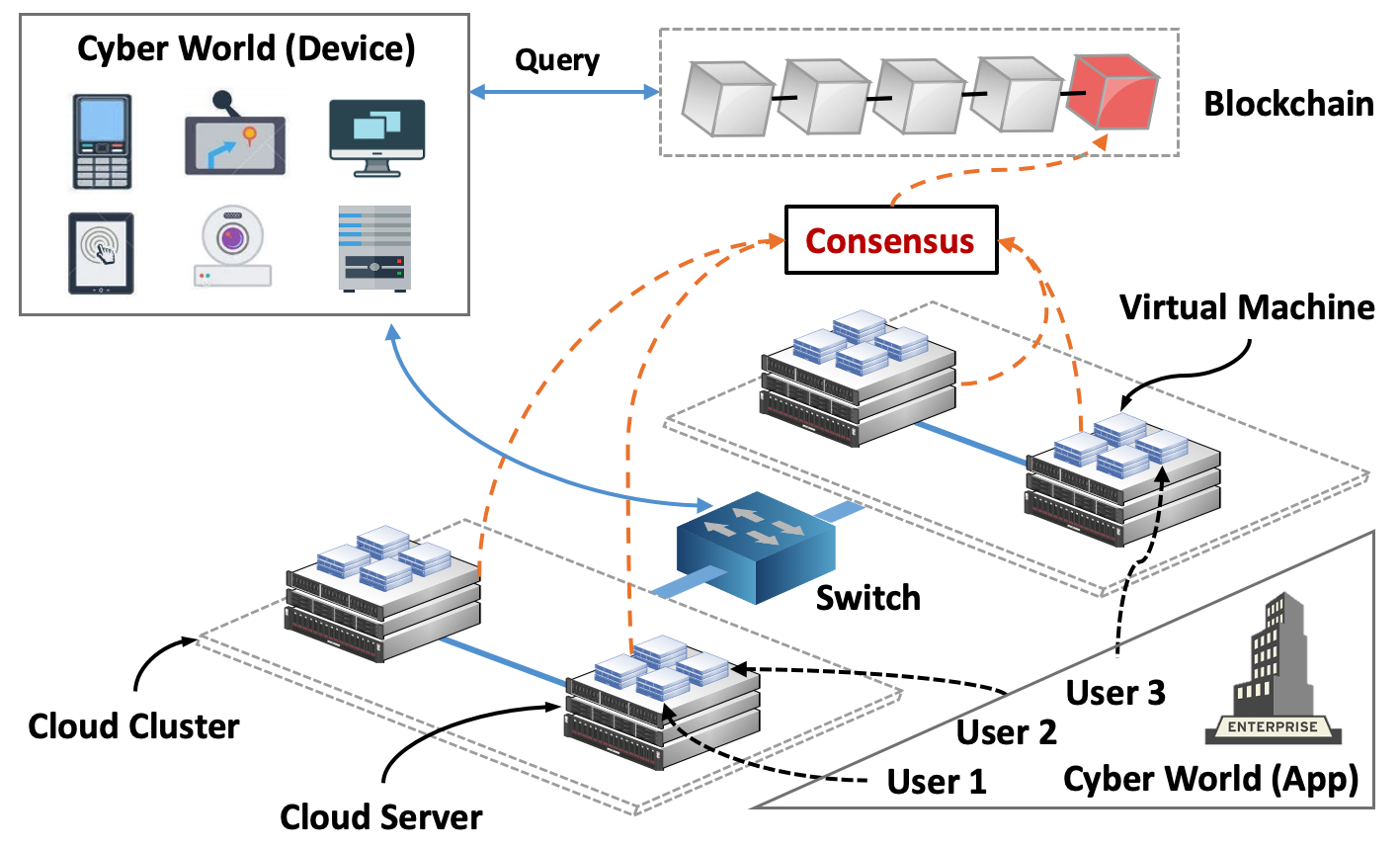}
	\caption{Infrastructure of cloud-oriented blockchains}
	\label{fig:cloud:blockchain}
\end{figure}

Fig.~\ref{fig:cloud:blockchain} demonstrates a typical application of permissioned blockchain in a cloud.  Such an architecture is the most popular one in industry for the purpose of offering blockchain-enhanced services to cloud users. 
Supported by cloud clusters,  each user can have at its disposal a virtual machine (VM) to join the blockchain network.  The VM can provide fast computing hardware,  on-demand storage,  high bandwidth\footnote{Up to 400 Gb/s Ethernet speed with the support of NVIDIA ConnectX Ethernet SmartNICs.},  and pre-loaded blockchain software so that users can participate in the consensus process to achieve agreement on the next state of the blockchain.  Users without VMs can also access the blockchain network and obtain data.  The cloud platform provides strong cloud firewalls, preventing unauthorized access to the blockchain network.  Leading cloud blockchain providers include the IBM blockchain platform \cite{IBMblockchain},   Mircrosoft Azure Blockchain Service \cite{AzureBlockchain},  and Amazon Managed Blockchain \cite{AWSBlockchain}, just to name a few. These platforms are built with open-source blockchain frameworks such as Hyperledger Fabric \cite{hyperledgerfabric2018} and Ethereum \cite{wood2014ethereum}. 

However, state-of-the-arts are just simple adoptions of traditional blockchains in clouds without deeply fusing blockchains with clouds,  thus fails to take advantages of cloud computing benefits. 
On one hand, existing  blockchains mostly adopt complicated mechanisms to address the partially-asynchronous communications inherent in Internet. Such schemes are based on the message-passing model that targets for loosely-coupled distributed systems where nodes communicate asynchronously by exchanging messages.  Nevertheless, cloud computing offers high-performance closely-coupled infrastructures to support reliable and synchronous communications\footnote{For synchronous communications,  we can assume that there exists a fixed upper bound $\Delta$ on the message transmission delay.}, which should significantly benefit consensus, but completely ignored by the existing popular blockchain algorithms. 
On the other hand,  traditional blockchain systems go though the network layers of an OS kernel, which contributes high latency to the block transmissions and consensus process, especially when the block size is large or the consensus algorithm has a complex design. Nevertheless, cloud platforms usually adopt the Remote Direct Memory Access (RDMA) technology that  bypasses the OS kernal, making the data transmissions extremely fast.
Therefore, fusing the blockchain technology seamlessly with cloud computing to make use of the above benefits has a great potential in building efficient permissioned blockchains, serving flexible business goals, and expanding blockchain applications to a vast new space. 

Nevertheless, designing a high-performance blockchain over clouds is non-trivial. Cloud-oriented blockchains can easily become centralized if relying on over-powerful authorities or validators who might completely control the consensus process. Such designs sacrifice decentralized trustlessness and make blockchains vulnerable to various threats. Additionally, even though cloud-oriented blockchains are mostly permissioned,  their security guarantees cannot be ignored, especially when considering Byzantine behaviors.  However, existing cloud-oriented blockchain implementations generally suffer from malicious behaviors, as analyzed in \cite{DBLP:conf/services/BrotsisKLBS20, davenport2018attack}.

Considering the above observations, we propose CloudChain, a cloud-oriented blockchain system, to address the following challenging questions: 1) how to take advantages of the closely-coupled nature of cloud computing platforms in building an efficient, robust,  and synchronous blockchain network; and 2) how to achieve low latency and high throughput, to avoid the centralization problem caused by over-powerful authorities, and to preserve security properties against strong Byzantine behaviors, in cloud-oriented blockchains.
CloudChain utilizes a shared-memory (SM) based consensus algorithm and RDMA technology.  It is a modularized three-layer system composed of the network layer, consensus layer, and blockchain layer.  The network layer provides a RDMA-enabled channel to support memory direct access and implement the shared memory model,  bypassing the  complex network protocols in OS kernels. In the consensus layer, we propose a novel SM-based consensus algorithm that can achieve fast consensus against Byzantine nodes.  The blockchain layer supports blockchain-related functionalities while ensuring persistence and liveness, the two critical blockchain properties for strong security guarantees. 

The primary contributions of this paper can be summarized as follows:
\begin{enumerate}
\item To our best knowledge, CloudChain is the first system based on a shared memory model that deeply integrates blockchain with cloud computing technologies.  In particular, we propose a shared memory based consensus algorithm to realize high-performance CloudChain.
\item CloudChain is a modularized blockchain system comprised of three layers: network layer, consensus layer, and blockchain layer, which together ensure the salient blockchain properties of decentralization, high throughput, persistence, and liveness.
\item To validate CloudChain, we carry out a thorough theoretical security analysis and build a CloudChain prototype for extensive experimental evaluations.  Our results shed light on the new designs of blockchains for cloud-oriented applications. 
\end{enumerate}

The rest of the paper is organized as follows.  Related works are presented in Section~\ref{sec:related:work}. Section~\ref{sec:background:motivation} introduces the background and our motivations. Section~\ref{sec:model} presents models and assumptions. In Section~\ref{sec:CloudChain:design}, we detail the CloudChain system. The security properties of CloudChain are analyzed in Section~\ref{sec:protocol:analysis}. We report the evaluation results in Section~\ref{sec:evaluation} and conclude this paper in Section~\ref{sec:conclusion}.

\section{Related Work}
\label{sec:related:work}

\textbf{Blockchain Protocol.} We summarize popular blockchain protocols that have appeared on the way of building scalable and efficient blockchain systems in chronological order. Bitcoin \cite{Bitcoin} and Ethereum \cite{wood2014ethereum}, as forerunners of blockchain technologies, started blockchain 1.0 and 2.0, respectively. Bitcoin-NG \cite{Bitcoin-NG} is an early attempt to scale Bitcoin from the perspective of the chain structure, in which the leader who generates a key block through mining can directly propose microblocks in subsequent rounds with a predefined rate until a new leader comes into play. ELASTICO \cite{elastico} is the first  to scale blockchain by sharding. It randomly partitions a blockcahin network into shards, with each running a Byzantine Fault-Tolerant (BFT) consensus, and then designates a final committee to reach consensus on one of the proposals collected from all shards. Hyperledger Fabric \cite{hyperledgerfabric2018}, as the most popular permissioned blockchain, leverages a message queue to reach consensus on blocks. The workflow of Hyperledger Fabric contains execution, ordering, and validation phases, and its kernel component is the Kafka protocol implemented in the ordering phase to determine the order of transactions. Algorand \cite{gilad2017algorand} adopts the cryptographic sortition algorithm to randomly select a committee based on the stake distribution, in which the committee members execute a BFT consensus algorithm, namely BA*, to determine the order of blocks. In CloudChain, we propose a novel shared-memory consensus algorithm that can securely and efficiently order transactions in clouds.

\textbf{Cloud-oriented blockchain.} There exist several attempts in industry to deploy blockchains based on the cloud computing technologies. 
IBM blockchain platform \cite{IBMblockchain} offers blockchain solutions based on Hyperledger Fabric and IBM cloud service.  Microsoft develops the Azure Blockchain Service \cite{AzureBlockchain} that allows users to operate and monitor blockchains on the Azure cloud platform. This service supports multiple types of blockchain systems including Ethereum, Hyperledger Fabric, and Cosmos. Amazon proposed the Amazon Managed Blockchain \cite{AWSBlockchain} based on the Amazon Quantum Ledger Database (QLDB). It supports Hyperledger Fabric and Ethereum.  
Effort in academia is quite limited. To our best knowledge,  BoR \cite{DBLP:journals/tsc/HuangJLZWTH20} proposed in 2020, a  permissioned blockchain designed for BaaS,  is the most relevant one, if not the only one for blockchain in clouds. BoR leverages the RDMA technology to optimize the EOS blockchain by accelerating the Delegated Proof of Stake (DPoS) consensus process while our CloudChain makes use of shared-memory as well as RDMA  to  deeply integrate  cloud and blockchain such that high-performance and strongly-secure blockchain services can be realized for cloud users.  
For further readings about cloud-oriented blockchains, we recommend two comprehensive surveys \cite{DBLP:journals/comsur/GaiGZY20, DBLP:journals/comsur/NguyenPDS20} to the interested readers. 

\textbf{Shared memory and RDMA technology.} The shared-memory model has been widely researched and it is shown to be computationally equivalent to the message-passing model in distributed computing \cite{DBLP:journals/jacm/AttiyaBD95}. However, shared memory might outperform message passing in efficiency and simplicity when deployed on cloud servers or datacenters. Shared-memory leader election and consensus algorithms have also been well studied. Bessani \textit{et al.} \cite{DBLP:journals/tpds/BessaniCFL09} proposed the Policy-Enforced Augmented Tuple Space (PEATS), which can be used in shared-memory systems to counter Byzantine behaviors. Aguilera \textit{et al.} \cite{DBLP:conf/podc/AguileraBCGPT18} presented the M\&M model that integrates shared memory and message passing. They provided an elegant proof showing that consensus algorithms based on M\&M model can tolerate $f>N/2$ crash failures. Another study of the M\&M model demonstrates that consensus algorithms can tolerate up to $N-1$ faulty processes and $N/2-1$ Byzantine processes \cite{DBLP:conf/podc/AguileraBGMZ19}. These works reveal the power of the shared-memory model. Besides, RDMA has been adopted by many high-performance distributed systems. DARE \cite{DBLP:conf/hpdc/PokeH15} builds replicated state machines leveraging RDMA. APUS \cite{DBLP:conf/cloud/WangJCYC17} is the first RDMA-based Paxos protocol using RDMA. To our best knowledge, no existing blockchain systems take advantage of RDMA and shared-memory to reach consensus while our CloudChain is the first one in this line of research.

\section{Background and Motivation}
\label{sec:background:motivation}

\subsection{Cloud-Oriented Blockchain}

Cloud-oriented blockchains are the ones implemented in cloud servers, benefiting both blockchain and cloud computing technologies. On one hand,  blockchain-enabled cloud computing can provide Blockchain-as-a-Service (BaaS) -- a provider that maintains blockchain services in clouds for individuals, companies, and government agencies. With BaaS, blockchain can be extended beyond its best-known applications in cryptocurrency to various businesses.  In addition,  it can empower cloud computing in a variety of areas, including data management, access control, privacy protection, and resource allocation.  
On the other hand, cloud computing can offer high-performance hardware, fast and reliable communication channels,  and secure runtime environments with strong firewalls, which can compensate for blockchains with a large amount of computational resources,  provide  synchronized networking environments for blockchains to quickly reach consensus,  and decreases the attack surfaces to blockchains.

\subsection{Message-Passing and Shared-Memory}
There exist two major communication models: message-passing and shared-memory in distributed computing. In message-passing, nodes communicate by exchanging messages while in shared-memory, nodes access a common (shared) memory through a set of predefined operations (e.g., read/write). These two models have been proved to be equivalent under certain assumptions, but each model has its own unique advantages compared to the other and benefits distributed systems in different ways \cite{DBLP:journals/jacm/AttiyaBD95}. 
Concretely,  the message-passing model is suitable for distributed environments where nodes reside on remote machines; therefore it is a top choice to build large-scale distributed systems. In contrast, the shared-memory model is commonly adopted to build tightly-coupled distributed systems that need to provide better performance for efficiency, synchrony, and fault-tolerance.  Nevertheless, shared-memory is generally facilitated with superior hardware support such as RDMA. 
Message-passing requires $f<N/2$ for the consensus problem to be solvable (i.e., a majority of honest nodes) while shared-memory can tolerate at most $N-1$ crash failures  in consensus \cite{DBLP:conf/podc/Abrahamson88},  where $N$ is the network size and $f$ is the maximum of faulty nodes.  A notable work by Aguilera \textit{et al.} \cite{DBLP:conf/podc/AguileraBCGPT18} presents an M\&M model that theoretically inherits salient properties of both message-passing and shared-memory models.  Despite the popularity of the shared-memory model in distributed computing,  it has not been introduced to blockchain systems.

\subsection{RDMA}

\begin{figure}[!htbp]
	\centering
	\includegraphics[width=3.5in]{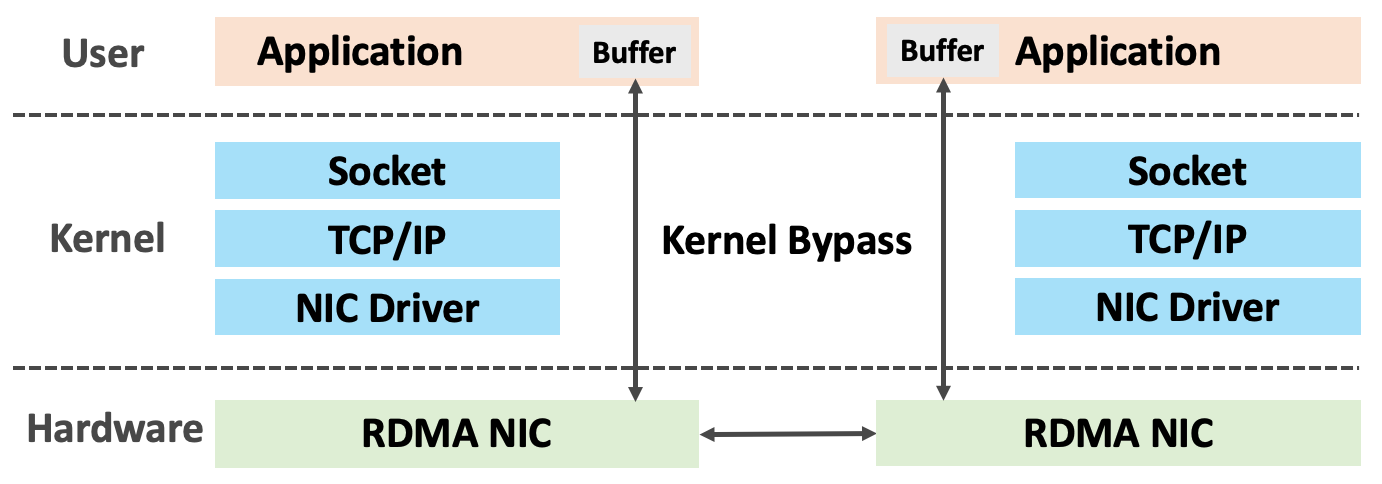}
	\caption{RDMA architecture layers}
	\label{fig:rdma}
\end{figure}

RDMA is a technology that was originally adopted by high-performance datacenters; it earns a spotlight in distributed systems in recent years.  As depicted in Fig.~\ref{fig:rdma},  RDMA allows computers to exchange data in memory without the kernel getting involved so that low latency and high throughput can be achieved.  RDMA has two notable advantages: 1) zero copy: data can be transferred directly from memory to memory though the RDMA Network Interface Controllers (NICs) of two nodes without burdening the CPU; 2) kernel bypass: no operating system and traditional networking stacks are involved during data transfers. With RDMA, data can be moved at low latency and CPU utilization. These salient features have attracted  studies of applying RDMA in distributed systems to address consensus problems \cite{DBLP:conf/podc/AguileraBGMZ19}. 
There exist three implementations of the RDMA technology: InfiniBand,  RDMA over Converged Ethernet (RoCE), and Internet Wide Area RDMA Protocol (iWARP). Particularly, RoCE is the most popular approach that can be used in Ethernet and it outperforms iWARP in terms of latency, throughput and CPU overhead. For a comparison study we refer the readers to \cite{mittal2018revisiting}.

In this paper we design CloudChain, a novel blockchain system that seamlessly integrates blockchian with the cloud computing technology.   CloudChain adopts a shared-memory model to establish a shared-memory consensus algorithm as well as a high-performance and secure blockchain system. To realize such an objective, it is intuitive to choose RDMA for constructing a shared memory, benefiting from its zero copy and kernel bypass features. In a nutshell, CloudChain is the first attempt to utilize the shared-memory model and RDMA technology in building a blockchain system.

\section{System Model and Assumption}
\label{sec:model}
In this paper, we assume that there is a set $V=\{v_1, v_2, \cdots, v_N\}$ of $N$ nodes, which are the virtual machines running  as validators who can participate in a consensus process.  Besides, all validators can communicate using an $N$-entry shared memory denoted as $\mathcal{M}=\{m_1, m_2, \cdots, m_N\}$. Each validator $v_i$ corresponds to a unique entry $m_i$, which can be manipulated through either local or remote operations.  A local operation can take effect instantaneously while the time delay of a remote operation is impacted by network conditions.  We implement $\mathcal{M}$ by the RDMA technology, which enables extremely high bandwidth and low latency. Therefore, it is reasonable to assume that 1) a network composed of honest validators is well connected with low latency -- it can be regarded as a synchronized network; and 2) each validator is equipped with a synchronized clock to precisely measure timeouts so that a known time delay $\Delta$ exists for any remote memory access. In our design, CloudChain proceeds by rounds, with each generating no more than one block, and its ledger structure is a chain of blocks as Bitcoin does \cite{Bitcoin}. 
In this paper, we denote frequently-used notations of transaction, block, block hash, and blockchain by $tx$, $B$, $h$, and $BC$, respectively, and use super/subscript to present more specific information.  For example, $B_k$ is the block generated at the $k$th round, and $B_0$ is the genesis block.  Each block $B_k$ except for $B_0$ is chained to the previous one by storing $h_{k-1}$, the block hash of $B_{k-1}$.  

Each validator can be either honest or Byzantine. Honest validators faithfully abide by a protocol while Byzantine validors can behave arbitrarily, deviating from the protocol, e.g., crash, fail-stop, corrupt, or collude. We denote Byzantine validators by adversary $\mathcal{A}$, which can corrupt its peers but a corruption  can be successful only after a short time period\footnote{This is reasonable since we always assume that an adversary is computationally limited, e.g., polynomial-time bounded}. The number of Byzantine validators is at most $f$ and the network size $N=2f+1$, i.e., no more than $1/2$ validators can be Byzantine.  
Each node can own a certain amount of digital coins and $\mathcal{A}$ controls less than 50\% of the coins of all validators.  We also assume that cryptographic primitives are secure in this paper.

\section{CloudChain Design}
\label{sec:CloudChain:design}

\subsection{Design Objectives}

CloudChain should guarantee the agreement on an immutable total order ledger among all validators, and achieve the following design objectives:
\begin{enumerate}
\item \textbf{Decentralization:} CloudChain should preserve decentralization, i.e., the (voting) power should not be controlled by centralized validators or authorities, even though it is implemented in a cloud.
\item \textbf{High throughput and low latency:} CloudChain needs to provide high throughput and low latency for reaching the consensus as well as confirming transactions. 
\item \textbf{Persistence and liveness:} CloudChain should guarantee the persistence and liveness properties against $f$ Byzantine validators.
\end{enumerate}
These objectives are straightforward except for the persistence and liveness properties,  which are formally defined as follows:
\begin{definition}{\textbf{Persistence and liveness.}}
\label{def:persistence}
   Persistence states that if an honest node proclaims a transaction, other honest nodes, if queried, either report the same result. Liveness, on the other hand, states that if an honest node generates a valid transaction, CloudChain eventually will add it to the blockchain.
\end{definition}

Rigorous descriptions of persistence and liveness are presented in the security analysis section.

\subsection{Layered CloudChain Architecture}
\label{sec:sub:architecture}

\begin{figure*}[!htbp]
	\centering
	\includegraphics[width=6.5in]{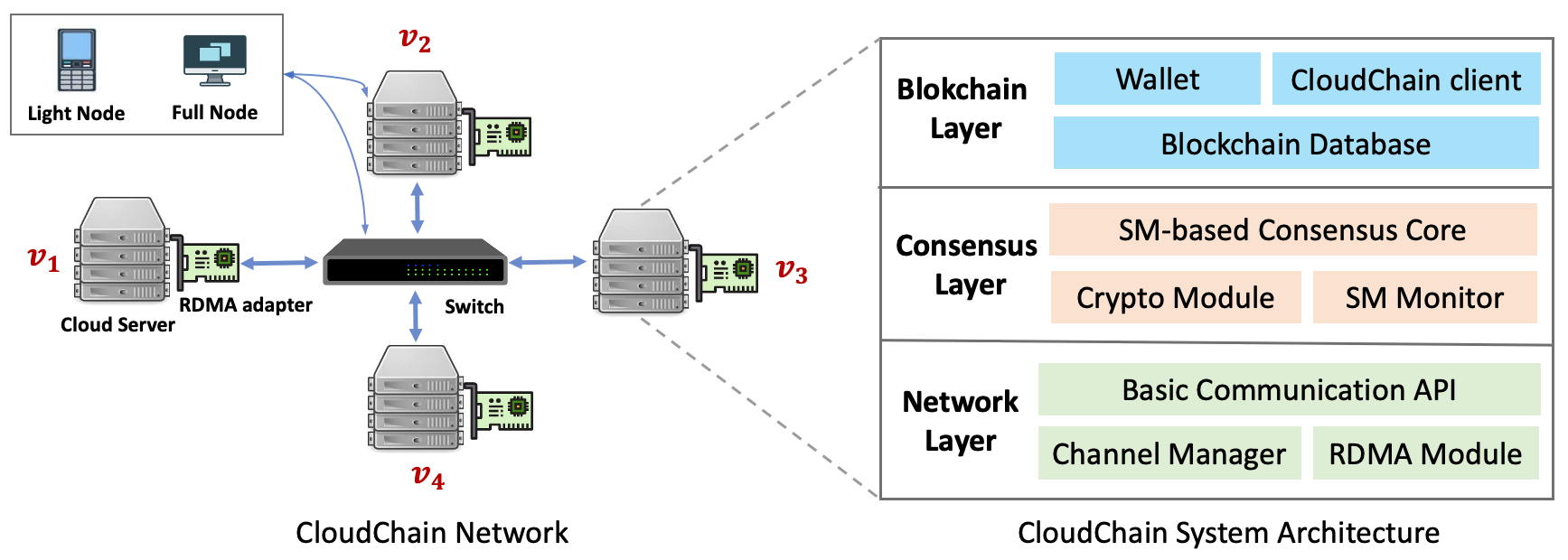}
	\caption{Network and system architecture of CloudChain}
	\label{fig:architecture}
\end{figure*}

\subsubsection{CloudChain Network}
Prior to the illustration of the CloudChain architecture, we explain what constitutes a CloudChain network. Specifically, three types of nodes compose the network: light nodes, full nodes, and validators.  
A light node runs a wallet and a CloudChain client whose details are presented in Section~\ref{sec:sub:blockchain:layer}.  It can connect to the CloudChain network, issue transactions, and store block headers\footnote{A chain of block headers is obtained by removing all block bodies of a blockchain} for a simplified payment verification. It is prohibited from participating in the consensus process. 
A full node, in contrast, locally stores a complete blockchain so that it can verify transactions and blocks;  thus full nodes are regarded as relay nodes that can validate received transactions and blocks, and then relay them to peers. They can serve light nodes by providing them with detailed information about a blockchain.  

A validator, which is a special full node,  participates in the consensus procedure. An eligible validator must meet three conditions beyond a full node:  1) as a validator, it should deposit a certain amount of money to obtain voting power, which is recorded by the blockchain; 2) validators must operate in a small-scale network with RDMA support, e.g., a datacenter or a cloud server where RDMA-enabled connections are stable; and 3) each validator must run the complete three-layer CloudChain system shown in Fig.~\ref{fig:architecture},  for  participating in a consensus process. 
These three types of nodes  can cover a wide range of users: light nodes are commonly run on lightweight devices, e.g., mobile devices; full nodes except for validators can be implemented on personal computers or edge servers; only validators need special RDMA support, and thus should be virtual machines at clouds. To better visualize how validators interconnect, we provide a simplified CloudChain network shown in Fig.~\ref{fig:architecture}, which consists of four validators $v_1, v_2, v_3, v_4$,  with each connecting to a switch using a RDMA adapter.

The CloudChain system contains three layers: network layer, consensus layer, and blockchain layer. The network layer provides RDMA support to realize the shared-memory communications.  Leveraging shared memory, the consensus layer offers a SM-based consensus core to determine the ordering of blocks. On top of the consensus layer, we build the blockchain layer to furnish CloudChain with storage, client, and wallet.  In the following subsections, we illustrate the layered architecture in detail.

\subsubsection{Network Layer}
\label{sec:sub:network:layer}
The network layer is comprised of three modules: RDMA, channel manager, and basic communication API, which together provide efficient and easy-to-use network functionalities.  The RDMA module is of the most basics in the network layer. It provides RDMA-enabled protocols to support the Queue Pair (QP) manipulations and low-level common user API.
A QP is similar to a socket in the context of the standard IP stack.  Specifically, a QP should be initialized on both sides to share communication ports so that all RDMA operations can be adopted.  
The low-level common user API includes \textit{VPI Verbs API}, \textit{RDMA\_CM API}, \textit{RDMA Verbs API},  and \textit{Events} (more details can be found in the Mellanox website \cite{Mellanox}). In this paper,  we employ the Paravirtual RDMA (PVRDMA) which offers better support for virtual machines.

Based on the RDMA module, we develop the channel manager and the basic communication API modules to better accommodate upper-layer uses of RDMA. The channel manager allows each node to manipulate the in-use channels. It serves two elementary objects: internal and external channels. An internal channel is set up for two nodes running on a single server (one server can serve multiple validators) while an external channel is established between two nodes on distinct servers. The channel manager is responsible for establishing, updating, and destroying both internal and external channels, as well as printing channel information. 

The basic communication API supports three modes, namely socket, unilateral, and bilateral.  Socket takes a traditional way of communicating through the TCP/IP protocol. CloudChain adopts sockets to transfer short messages such as those carrying user identities, transactions, memory addresses,  and QP information. The unilateral communication mode provides \textit{rdma\_read} and \textit{rdma\_write} operations, with the former aiming to fetch data from a memory region and keep it in a buffer while the latter writing data into a memory region. Note that directly writing data into remote memories may result in severe security problems; thus such operations are disabled and we enforce an access policy that requires a registered node to write only on its own local memory but can read both local and remote memories. Finally, the bilateral communication mode involves a couple of operations, namely \textit{rdma\_send} and \textit{rdma\_receive}, which should be executed in pairs. Specifically, when a node $v$ sends data to a designated destination $u$,  one should ensure that $u$ has issued a corresponding receive request.

\subsubsection{Consensus Layer}
\label{sec:sub:consensus:layer}
The consensus layer consists of a crypto module, an SM monitor, and an SM consensus core. The crypto module provides a standard crypto library with various cryptographic schemes, of which CloudChain adopts the authenticated encryption scheme, hash function, and elliptic curve cryptography. Our crypto module is built with the Crypopp \cite{Crypopp}. 

To better support the consensus layer, we equip CloudChain with the SM monitor that can operate on shared memory with three basic functions: $\mathtt{Write}()$, $\mathtt{Read}()$, and $\mathtt{Scan}()$, which are realized with the support of network-layer functionalities and are designed to be directly called by the SM consensus core.  
Specifically, following the access rules formulated by the basic communication API, $\mathtt{Write}(m_i, \langle \cdot \rangle)$ can write data into $m_i$ only if the executant is exactly $v_i$ while $\mathtt{Read}(m_i, \langle \cdot \rangle)$ can locally or remotely read data from $m_i$, where $\langle \cdot \rangle$ stands for unspecified data content.
The $\mathtt{Scan}(\mathcal{M})$ function collects data from the memory regions $\mathcal{M}$. Since data might appear with time differences,  we need to scan $\mathcal{M}$ multiple times to  ensure a complete collection. To reduce the involved communication overhead, $\mathtt{Scan}()$ only reads the data not yet present during the current round. To achieve this functionality, $\mathtt{Scan}()$ is furnished with a detection mechanism to track changes of peers' local memories using a scan map. This method reduces the communication complexity to $O(N)$ despite the multi-round scanning. More details of $\mathtt{Scan}()$ are further demonstrated in Section~\ref{sec:sub:consensus}.

The SM-based Consensus Core is of the most importance in CloudChain since it determines the ordering of blocks,  as well as the efficiency and security of CloudChain.  We build it as an efficient BFT middleware that coordinates nodes to reach consensus on a block proposal.  Concretely, the consensus core contains the proposer selection queue, transaction pool, state synchronization tool, and SM-based Consensus Algorithm (SMCA), with all working together to achieve an agreement on the next block.
The idea of the proposer selection queue was initially proposed by \cite{tendermint}, targeting to maintain a deterministic queue of leaders, by which all validators can query who is the leader of the current round.  The queue is updated after each round by letting validators move ahead in the queue according to their voting powers.  
The transaction pool provides each validator with a stack to store valid transactions received from peers for transaction collection. A validator selected as a leader can serialize transactions and pack them into a new block. The state synchronization tool helps nodes synchronize their ledgers with peer nodes, which is also responsible for recovery. SMCA intends to determine the next block while protecting against at most $f$ Byzantine nodes. We put details of SMCA in Section~\ref{sec:sub:consensus}.

\subsubsection{Blockchain Layer}
\label{sec:sub:blockchain:layer}
The blockchain layer integrates blockchain-related components including wallet, CloudChain client, and blockchain database. In a wallet, users can securely store coins and acquire their account balances and information of approved transactions. Note that a wallet does not rely on the network and consensus layers; thus can be run by any user who intends to interact with CloudChain.  Each user is assigned with a unique identity,  which can be used to register multiple accounts.

We set up a CloudChain client based on a Remote Procedure call (RPC) client/server model and the RDMA-enabled socket. A user can run a CloudChain client to remotely connect to a full node.  A full node can handle requests to service light nodes such as requiring information of a certain block.  A light node equipped only with a CloudChain client can issue an instant transaction but cannot participate in the consensus process since only validators implemented with a three-layer system stack have the right to propose or vote on a block. The blockchain database stores the immutable blockchain ledger organized as a chain of blocks, as well as the state information including the current block height, validator list, and leader identity. 
The design motivation of the modularized blockchain layer is to serve various users with heterogeneous hardware capacity and let high-performance validators with RDMA support execute the consensus process.

\subsection{Initialization}
\label{sec:sub:init}

\begin{algorithm}[!htbp]
	\label{alg:init}
	\DontPrintSemicolon
	\caption{Initialization Procedure}
	\SetKwInOut{Input}{input}
	\SetKwProg{Fn}{Procedure}{}{}
	\Fn{Init($v_i$)}{
		$\triangleright$ \textcolor{SpecialBlue}{identity establishment}\;
		$(\sk_i, \pk_i) \leftarrow \kgen(\secparam)$\;
		$id_i=\pk_i$, and broadcast $id_i$\;

		$\triangleright$ \textcolor{SpecialBlue}{shared memory connection}\;
		\For {each $m_j \in \mathcal{M}$}{
			ConnectQP($m_i, m_j$)\;
		}

		$\triangleright$ \textcolor{SpecialBlue}{initialization for $\mathtt{Scan}(\mathcal{M})$ }\;
		create $scan\_map([\vec{id}, \vec{b}])$\;

		$\triangleright$ \textcolor{SpecialBlue}{blockchain initialization}\;
		initialize the proposer selection queue $Q$ to obtain the first leader\;
		$BC_0=B_0$\;
	}
\end{algorithm}

Let's first examine the initialization procedure executed by each node in the outset. This procedure consists of four stages, namely identity establishment, shared memory connection,  $\mathtt{Scan}()$ initialization, and blockchain initialization.
The key-generation algorithm $\kgen(\secparam)$ takes as input the security parameter $\secparam$ written in unary, and generates a key pair $(\sk_i, \pk_i)$ for node $v_i$, where $\sk_i$ is a private key and $\pk_i$ is a public key. Then $v_i$ regards $\pk_i$ as its identity $id_i$ and broadcasts $id_i$ to the CloudChain network. To establish shared memory connections with the other nodes, $v_i$ calls ConnectQP($m_i, m_j$) for each $m_j \in \mathcal{M}$ to create a QP between $m_j$ and $m_i$. Afterwards, $v_i$ constructs a scan map represented by $scan\_map([\vec{id}, \vec{b}])$, where $[\vec{id}, \vec{b}]$ denotes $\{(id_j, b_j)|j=1,2,\cdots, N\}$, mapping each $id$ to a binary variable $b$. Note that $scan\_map([\vec{id}, \vec{b}])$ is prepared for the $\mathtt{Scan}()$ algorithm illustrated in Section~\ref{sec:sub:consensus}. Then $v_i$ initializes its proposer selection queue $Q$ by computing the first leader according to the strategy in \cite{tendermint}.  Finally, $v_i$ initializes its blockchain with the genesis block $B_0$,  described as $BC_0=B_0$. 

\subsection{SM-based Consensus Algorithm (SMCA)}
\label{sec:sub:consensus}
As the core of CloudChain, SMCA determines how to append new blocks orderly and securely. In this subsection, we present the details of SMCA to show how to achieve BFT consensus with shared memory. As shown in Algorithm~\ref{alg:consensus}, SMCA proceeds by three phases: $\texttt{PROPOSE}$, $\texttt{COMMIT}$, and $\texttt{DECIDE}$. 

In the $\texttt{PROPOSE}$ phase, each validator has access to $Q$ that determines the leader of the current round.  During the $k$-th round,  a validator $v_i$ who executes $\mathtt{Proposer}(Q, k)$ queries $Q$ to obtain the identity $id_l$ of leader $v_l$.  If $id_i=id_l$, $v_i$ appears as the leader for the $k$-th round; otherwise, it becomes a follower and recognizes $v_l$ as the leader. Afterwards, the leader assembles the new block $B_k$ and write the proposal $\langle\mathtt{PROPOSE}, B_k, \sigma\rangle$ into $m_i$, where $\sigma$ is the signature of the concatenation of $\mathtt{PROPOSE}$ and $B_k$. Then each follower reads the leader's memory $m_l$ to fetch the proposal. To prevent endless waiting for a proposal,  each follower sets a timeout so that it can jump to Line 39 of Algorithm~\ref{alg:consensus} to abandon the current round if $\mathtt{Time}()\geq T_0+\delta_1$, where $\mathtt{Time}()$ returns the current UNIX timestamp, $T_0$ is the termination time of the previous round, and $\delta_1$ is a constant that can be adjusted according to specific implementations. 

In the $\texttt{COMMIT}$ phase, each validator constructs two sets $S_0[h]$ and $S_1[h]$ to store votes, which are initialized to be empty. In particular, $S_0[h]$ is the set of validators proposing FALSE for the current block whose hash is $h$ while $S_1[h]$ is the set of validators proposing TRUE for $h$.
Then each validator verifies the $\sigma$ and $B_k$ parsed from the proposal obtained in the $\texttt{PREPARE}$ phase.  \emph{Verify}($B_k$) outputs TRUE if $B_k$ has valid transactions and a correct block header. Only when $\sigma$ is valid and \emph{Verify}($B_k$)==TRUE can $v_i$ write a commit message $\langle\mathtt{COMMIT}, h_k, \text{TRUE}, \sigma \rangle$ into $m_i$ and add itself to $S_1[h_k]$; otherwise, $v_i$ writes $\langle\mathtt{COMMIT}, h_k, \text{FALSE}, \sigma \rangle$ into its memory and inserts itself to $S_0[h_k]$. 

\begin{algorithm}[!htbp]
	\label{alg:scan}
	\DontPrintSemicolon
	\caption{Memory Scan Algorithm}
	\SetKwInOut{Input}{input}
	\SetKwProg{Fn}{Function}{}{}
	\Fn{$\mathtt{Scan}(\mathcal{M})$}{
		Initially, create \emph{buffer}[] and $count:=0$\;
		\For {each $(id_j, b_j)$ in $scan\_map(\vec{id}, \vec{b})$}{
			\If {$b_j==0$}{
				$\mathtt{Read}(m_j, \langle \cdot \rangle)$\;
				\If {read operation is successful}{
					add $(id_j, \langle \cdot \rangle)$ to buffer[$count$++]\;
					$b_j=1$\;
				}
			}
		}
		return \emph{buffer}[]\;
	}
\end{algorithm}

Before proceed, we need to explain the memory scan algorithm $\mathtt{Scan}(\mathcal{M})$ described in Algorithm~\ref{alg:scan}. $\mathtt{Scan}(\mathcal{M})$ looks through $\mathcal{M}$ but only reads data from the entries that have not presented votes in the current round. To achieve this functionality,  $\mathtt{Scan}(\mathcal{M})$ tracks changes of $\mathcal{M}$ by maintaining $scan\_map([\vec{id}, \vec{b}])$. If $b_j==0$, which means $v_j$ has not yet voted, $\mathtt{Scan}(\mathcal{M})$ keeps on reading data from $m_j$. Once a read operation is successful, $(id_j, \langle \cdot \rangle)$ is added to the buffer and $b_j$ is set to $1$. When the \emph{for} loop ends, $\mathtt{Scan}(\mathcal{M})$ returns the buffer. After each round, all entries in $\vec{b}$ are reset to zero. 

Now we are back to continue explaining the $\texttt{COMMIT}$ phase (Line 22-34).  For each commit message read from the returned buffer, $v_i$ verifies $\langle\mathtt{COMMIT}, h_k, \textit{VOTE}, \sigma \rangle$ received from $v_j$ by checking $\sigma$,  where \emph{VOTE} can be either TRUE or FALSE. If the validation succeeds, $v_i$ adds $\{v_j\}$ to $S_1[h_k]$ only if \textit{VOTE} is TRUE; otherwise adds $\{v_j\}$ to $S_0[h_k]$. Then for each $h_k$, if $|S_0[h_k]|\geq f+1 \vee |S_1[h_k]|\geq f+1$, $v_i$  ends the \emph{while} loop and steps into the $\texttt{DECIDE}$ phase. Note that the bound $f+1$ is adopted to defend Byzantine behaviors since there are at most $f$ Byzantine nodes and $f+1$ nodes are sufficient to cover at least one honest node. Finally, $v_i$ should abandon the current round and goto Line 39 if $\mathtt{Time}()\geq T_0+\delta_2$, where $\delta_2$ is an adjustable constant as $\delta_1$ functions.

In the $\texttt{DECIDE}$ phase, each validator checks whether $|S_1[h_k]|\geq f+1$, meaning that more than $f$ validators have voted an \emph{accept} for $B_k$. If such a condition holds, $v_i$ appends $B_k$ to $BC_{k-1}$; otherwise, $v_i$ discards $B_k$.  Then each validator increments $k$ and enters into the next round with $T_0=\mathtt{Time}()$.

\begin{algorithm}[!htb]
\label{alg:consensus}
\DontPrintSemicolon
\caption{SM-based Consensus Algorithm ($k$-th round)}
$\triangleright$ \textcolor{SpecialBlue}{ \textbf{\texttt{PROPOSE}}}\;
$id_l=\mathtt{Proposer}(Q, k)$\;
\If {$id_i==id_l$}{
	$\triangleright$ \textcolor{BrickRed}{as a leader}\; 
	$\mathtt{Write}(m_i, \langle\mathtt{PROPOSE}, B_k, \sigma\rangle)$\;
} \Else {
	$\triangleright$ \textcolor{BrickRed}{as a follower}\; 
	\While{$\mathtt{Time}()<T_0+\delta_1$}{
		\If {$\mathtt{Read}(m_{l})$ is successful}{
			end the \emph{while} loop and goto line 12\;
		}
	}
	abandon the current round and goto Line 39\;
}

$\triangleright$ \textcolor{SpecialBlue}{ \textbf{\texttt{COMMIT}}}\;
$\triangleright$ \textcolor{BrickRed}{for each validator}\; 
$S_0[h]=\emptyset$ $\triangleright$ set of validators proposing FALSE\;
$S_1[h]=\emptyset$ $\triangleright$ set of validators proposing TRUE\;
\If {$\sigma$ is valid and Verify($B_k$)==TRUE}{
	$\mathtt{Write}(m_i,  \langle\mathtt{COMMIT}, h_k, \text{TRUE}, \sigma \rangle$)\;
	$S_1[h_k]=S_1[h_k]\cup \{v_i\}$\;
} \Else {
	$\mathtt{Write}(m_i,  \langle\mathtt{COMMIT}, h_k, \text{FALSE}, \sigma \rangle$)\;
	$S_0[h_k]=S_0[h_k]\cup \{v_i\}$\;
}
\While {TRUE}{
	\emph{buffer}[]$\leftarrow \mathtt{Scan}(\mathcal{M})$\;
	\For {each commit message read from buffer[]}{ 
		\If {$\langle\mathtt{COMMIT}, h_k, \text{VOTE}, \sigma \rangle$ from $v_j$ is valid }{
			\If {VOTE==TRUE}{
				$S_1[h_k]=S_1[h_k]\cup \{v_j\}$\;
			} \Else {
				$S_0[h_k]=S_0[h_k]\cup \{v_j\}$\;
			}
		}
	}

	\For {each $h_k$}{
		\If {$|S_0[h_k]|\geq f+1 \vee |S_1[h_k]|\geq f+1$} {
			end the \emph{while} loop and goto the \texttt{DECIDE} phase\;
		}
	}

	\If {$\mathtt{Time}()\geq T_0+\delta_2$}{
		abandon the current round and goto Line 39\;
	}
}
$\triangleright$ \textcolor{SpecialBlue}{ \textbf{\texttt{DECIDE}}}\;
$\triangleright$ \textcolor{BrickRed}{for each validator}\; 
\If {$|S_1[h_k]|\geq f+1$}{
	$BC_k=BC_{k-1}+B_k$\;
}
enter the next round with $k=k+1$, $T_0=\mathtt{Time}()$\;
\end{algorithm}

\section{Protocol Analysis}\label{sec:protocol:analysis}

In this section, we provide a rigorous analysis on the security strength of CloudChain, i.e., CloudChain possesses the persistence and liveness properties.

\begin{theorem}{Persistence.}
\label{thm:persistence}
   If an honest node $v_i$ proclaims a transaction $tx_k^t$, the $t$-th transaction in the $k$-th block, other nodes, if queried, should report the same result.
\end{theorem}
\begin{proof}
Proving persistence is equivalent to examining whether for any two blockchains $BC^i$ and $BC^j$ respectively owned by nodes $v_i$ and $v_j$, there exist two different transactions $tx^i\in BC^i$ and $tx^j\in BC^j$ that are in the same position as $tx_k^t$ and respectively contained by $B_k^i$ and $B_k^j$. To prove by contradiction, we assume that such $tx^i$ and $tx^j$ do exist, and there are two cases to consider when this assumption holds. 

Case 1: $tx^i$ and $tx^j$ are respectively appended to the blockchains $BC^i$ and $BC^j$ in the same $k$-th round. This case might diverge a blockchain and make it vulnerable to double-spending attacks. Under this circumstance, a Byzantine leader can propose the two different blocks $B_k^i$ and $B_k^j$ in the $\texttt{PROPOSE}$ phase of the same round. Such a leader may achieve this goal by rapidly substituting $B_k^i$ in its memory $m_l$ with $B_k^j$ so that nodes read $m_l$ might obtain different blocks due to time differences. However, CloudChain certainly can prevent from accepting both. It is important to observe that in CloudChain each validator should put $h_k$ into $\langle\mathtt{COMMIT}, h_k, \text{FALSE}, \sigma \rangle$, and each honest node only commits on one block, either $B_k^i$ or $B_k^j$. The case when both are accepted indicates that there are at least $f+1$ nodes voting for $B_k^i$ and another $f+1$ nodes voting for $B_k^j$, which implies $N\geq 2f+2$, contradicting the fact that $N=2f+1$. 

Case 2: $tx^i$ is appended to the blockchain $BC^i$, but at the same round $k$, $BC^j$ does not accept any block. In this case, the leader might not be Byzantine and proposes one block. However, the blockchain might diverge when some nodes decide on accepting the block while some other nodes decide on abandoning it. It is straightforward to observe that such a harmful case cannot happen since it requires at least $f+1$ nodes voting for accepting $B_k^i$ and another $f+1$ nodes voting for discarding $B_k^i$,  contradicting the fact that $N=2f+1$.

In a nutshell, all the nodes queried for a transaction should report the same result or report error messages. 
\end{proof}

\begin{theorem}{\textbf{Liveness.}}
\label{thm:liveness}
  If an honest node issues a valid transaction and broadcasts it, CloudChain adds it to the blockchain within $T$ rounds w.h.p.\footnote{In this paper,  we say that an event $E$ occurs with high probability (w.h.p.) if for any $c\geq1$, $E$ occurs with probability at least $1-1/N^c$. }, where $T$ is a sufficiently large integer.
\end{theorem}

\begin{proof}

Before proving a valid transaction can be accepted within $T$ rounds, we demonstrate that CloudChain is free from deadlocks. Since our network is synchronous, the communication overhead of each SM-related operation executed by an honest node is bounded by a very small constant $\Delta$, based on which we set up two timeouts for the $\texttt{PROPOSE}$ and $\texttt{COMMIT}$ phases to prevent deadlocks caused by Byzantine behaviors. A Byzantine leader might not present its proposal  in the $\texttt{PROPOSE}$ phase; however, followers can abandon the current round if they read nothing from $m_l$ after waiting for $\delta_1$. A validator might not be able to collect $f+1$ votes to confirm a decision  in the $\texttt{COMMIT}$ phase. In such a case, the second timeout $T_0+\delta_2$ is introduced to avoid endless waiting. The $\texttt{DECIDE}$ phase contains no loops so that deadlock cannot happen during this phase. Therefore, CloudChain must proceed no matter what Byzantine behaviors happen. Roughly, the one-round execution has to be completed within $\delta_2$. 

Next we formalize the growth of CloudChain as a Possion process. Assume an honest node issues and broadcasts a transaction $tx$. Other honest nodes can receive $tx$ within $\Delta$. Since the stake ratio of Byzantine nodes is less than $1/2$, the probability that an honest node is selected to propose a block in a sufficient long interval of length $T$ is larger than $1/2$. Let the total number of honest transactions (malicious transactions are discarded by honest nodes when packed into a block) be $\widetilde{T_N}$ and the throughput be $\widetilde{T_H}$.  Since an honest leader first packs older transactions into a block, $tx$ is added to blockchain using $\widetilde{T_N}/\widetilde{T_H}$ good rounds. Let $X_i$ denote the event where the $i$th round is good.  Let $X^T=\sum_1^T X_i$ and $\mu=T/2$.  We have $\mathbb{E}[X_i]\ge T/2$. Applying the Chernoff bound, we obtain
\begin{equation}
    Pr[X^T\leq (1-\delta)\mu]\leq e^{-T/16},
\end{equation}
where $\delta = 1/2$. Therefore,
\begin{equation}
    Pr[X^T> T/4]=1-Pr[X_T\leq T/4]>1-e^{-T/16}.
\end{equation}

This indicates that when $T$ is sufficiently large, there is at least $T/4$ good rounds w.h.p. On the other hand, it is trivial to satisfy the requirement that $\widetilde{T_N}/\widetilde{T_H}<T/4$ when $T$ is sufficiently large. In other words, CloudChain requires that $T>4\widetilde{T_N}/\widetilde{T_H}$. This completes the proof. 
\end{proof}

\section{Evaluation}
\label{sec:evaluation}

\begin{table}[htbp]
	\caption{Implementation Specifications}
	\begin{center}
	\noindent\begin{tabular}{m{4cm}<\centering  m{4cm}<\centering}
	\toprule
	\textbf{Specifications} & \textbf{Value} \\ 
	\midrule
	Device & PowerEdge R720xd \\ 
	Vendor & Dell \\
	CPU Cores &	16 \\
	CPU Threads & 32 \\
	RAM	& 16 GB DDR3 \\
	OS & VMware ESXi-7.0U1c \\
	OS on Virtual Machine & Ubuntu20.04 \\
	RDMA NIC & LRES1004PF-2SFP+ \\
	RDMA Support & PVRDMA\\

	\bottomrule
	\end{tabular}
	\end{center}
	\label{tb:implementation}
\end{table}

\textbf{Configuration:} We implement CloudChain with 3948 lines of C/C++ code and conduct the experiments on a RoCEv2-based testbed with two DELL PowerEdge R720xd servers (denoted as server A and B). As shown in TABLE~\ref{tb:implementation}, each server is equipped with two Intel Xeon E5-2650 CPU processors, each having 16 CPU cores and 2.00 GHz frequency.  The DRAM size in each server is 16 GB (4$\times$4 GB) with the type of DDR3, whose speed is 1333 MHz. The L1 cache of each CPU core is 256 KB while the L2 cache is 2 MB. There are 16-cores on one CPU processor to share the same 20 MB L3 cache. The RNIC in our servers is a LRES1004PF-2SFP+,  whose speed is 10000 Mbps. The OS release of all servers is VMware ESXi-7.0U1c. VMware vCenter Server (Vcenter 7.0.1) is installed on each physical server to manipulate multiple virtual machines, with each occupying about 1 GB of RAM and 10 GB of disk space.  The OS release of the virtual machine is Ubuntu 20.04.  In each virtual machine, a virtual Peripheral Component Interconnect express (PCIe) device provides Ethernet interface through the Paravirtual RDMA (PVRDMA) which is a virtual NIC, supporting the standard RDMA API. Except where noted, we measure 200 times for each node to obtain the average of each data point after a warm-up of 30 seconds. Our evaluation quantitatively answers the following questions:
\begin{enumerate}
\item What is the latency of CloudChain for each  basic operation involved in QPs, commit messages, and blocks? (Section~\ref{sec:evaluation:basic})
\item What is the latency of SMCA with variant network size and Byzantine ratio? what are  the throughput and latency of CloudChain when network size and block size are varying? (Section~\ref{sec:evaluation:CloudChain})
\end{enumerate}

\subsection{Performance of Basic Operations}
\label{sec:evaluation:basic}

\textbf{QP (Connection and Disconnection).} We first measure the latency of the connection and disconnection of QPs and report the results in Fig.~\ref{fig:qp}.  There exist 8 and 7 VMs respectively located at server A and B.  Each node needs to create two QPs for the memory regions used in transmitting commit messages and blocks.  For simplicity,  we abbreviate the term \emph{machine} as ``MACH'' in this section, and ``1 MACH'' means the QP is maintained by two nodes in server A while ``2 MACHs'' indicates two nodes are located in distinct servers.  On server A,  it takes about 8.304 ms and 4.786 ms to connect and disconnect a QP for a commit region while about 8.473 ms and 5.034 ms for a block region.  With distinct servers, it takes about 6.022 ms and  5.934 ms to connect and disconnect a QP for a commit region, and about 6.943 ms and 6.075 ms for a block region. Establishing connections for commit and block regions take a different amount of time for both cases of ``1 MACH'' and ``2 MACH'' because each connection involves the partition of the memory spaces located to commit or block, which are different, though connection itself takes the same amount of time for each case.   Nevertheless, disconnections that are free from communications have a nearly identical latency since the time consumed by freeing memory regions is not obviously impacted by the memory size.
Fig.~\ref{fig:qp} also reveals that when nodes are located on a single server, connecting and disconnecting a QP cost a longer time.  This is presumably because when two or more VMs reside on a single server,  they share the computational resources in that machine \cite{vSphere} causing the performance of manipulating QPs slightly reduced.  

\begin{figure}[!htbp]
	\centering
	\includegraphics[width=3.5in]{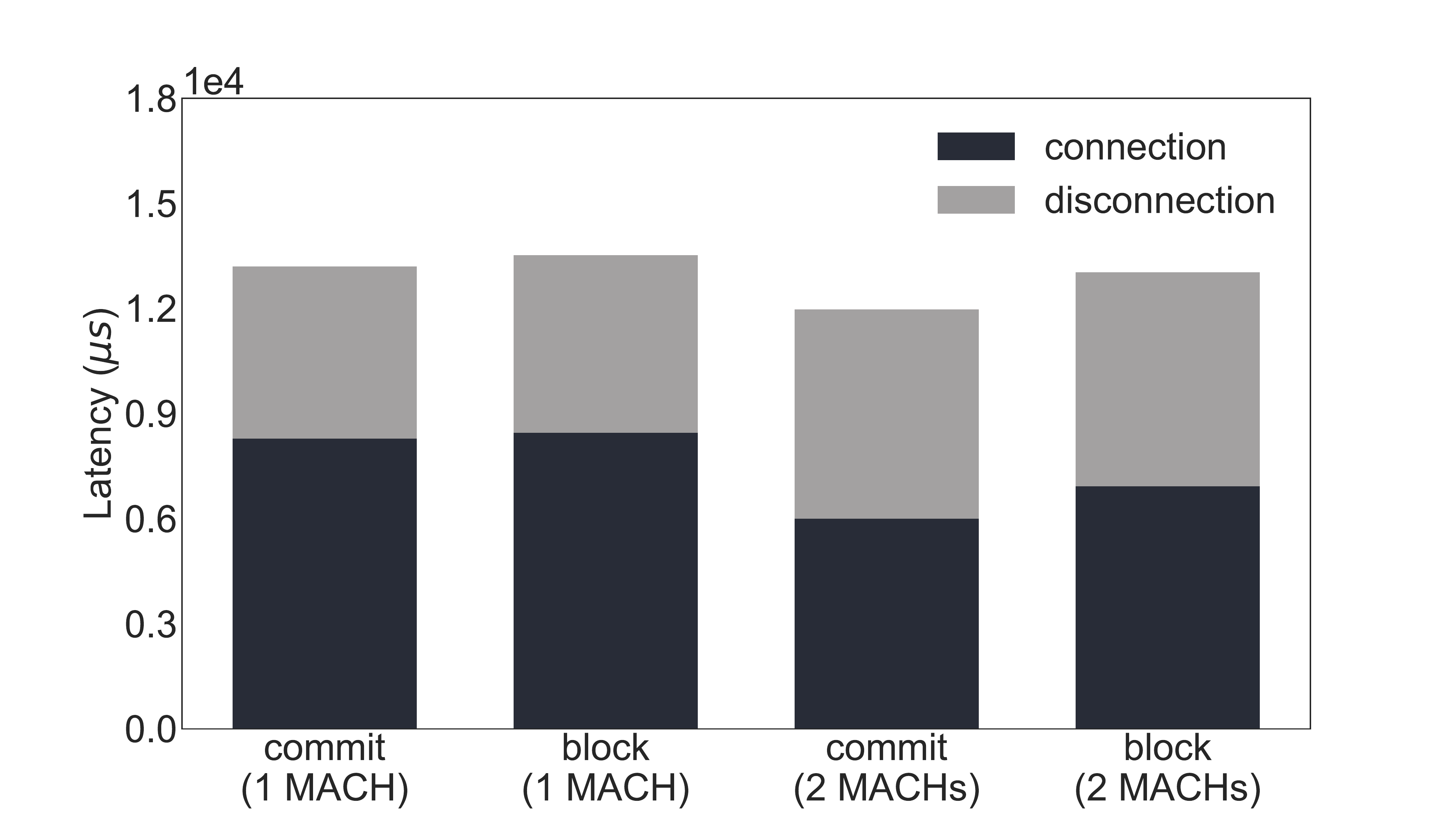}
	\caption{Latencies for QP connection and disconnection}
	\label{fig:qp}
\end{figure}

\begin{figure}[!t]
	\centering
	\subfigure[Send/Receive (1 MACH)]{
		\label{fig:sr}
		\centering
		\includegraphics[width=0.23\textwidth]{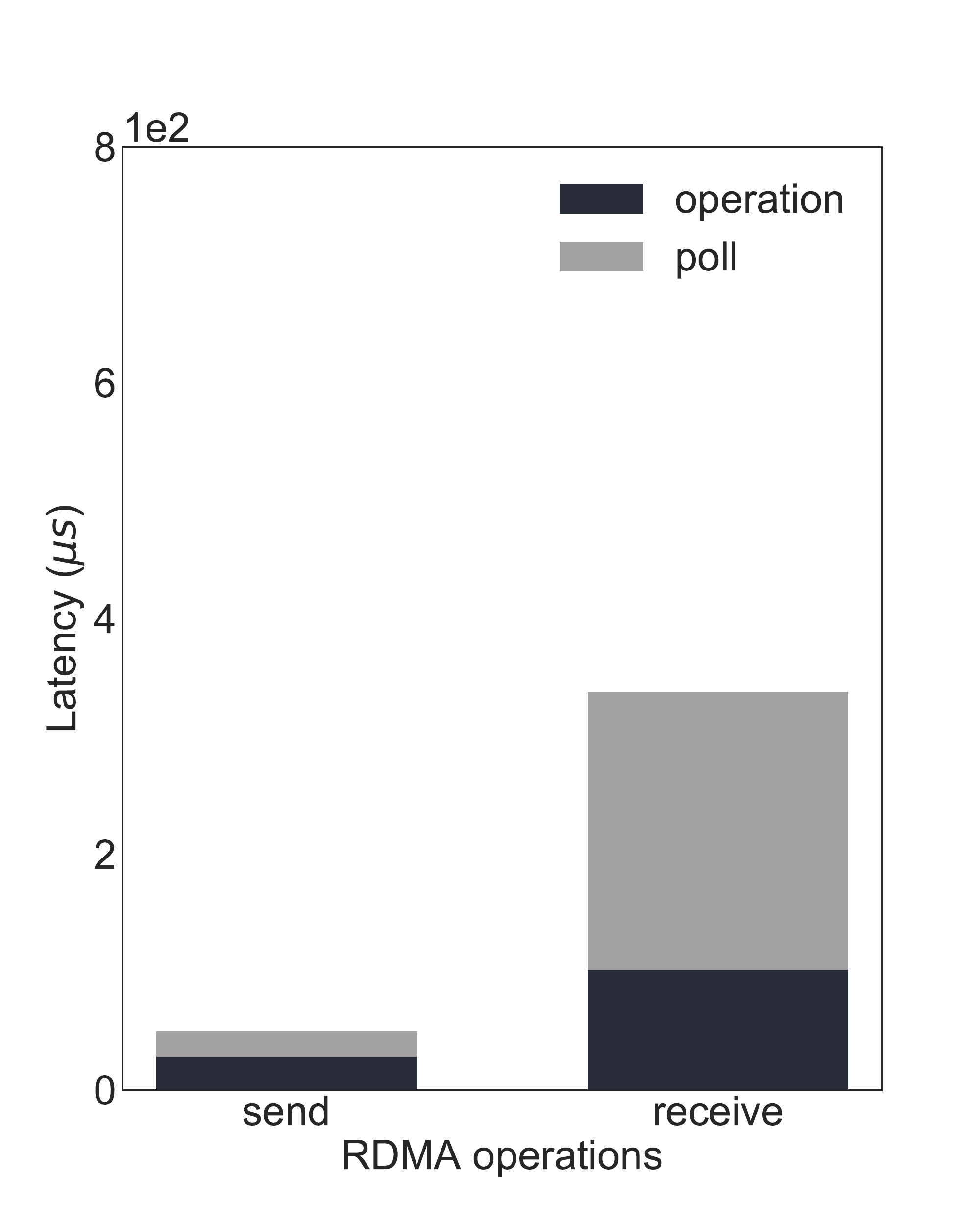}
	}
	\subfigure[Read/Write (1 MACH)]{
		\label{fig:rw}
		\centering
		\includegraphics[width=0.23\textwidth]{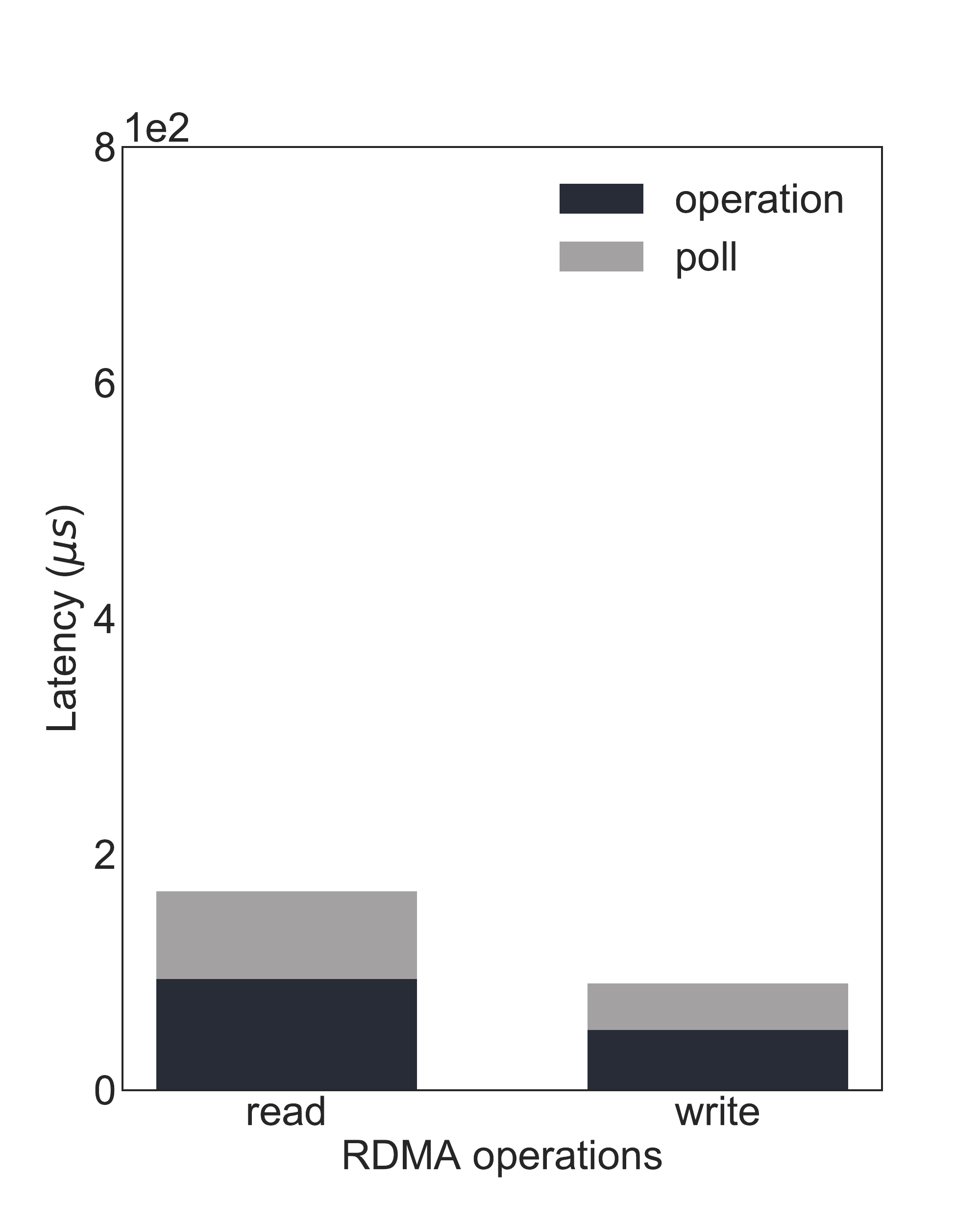}
	}
	\subfigure[Send/Receive (2 MACHs)]{
		\label{fig:sr:dif}
		\centering
		\includegraphics[width=0.23\textwidth]{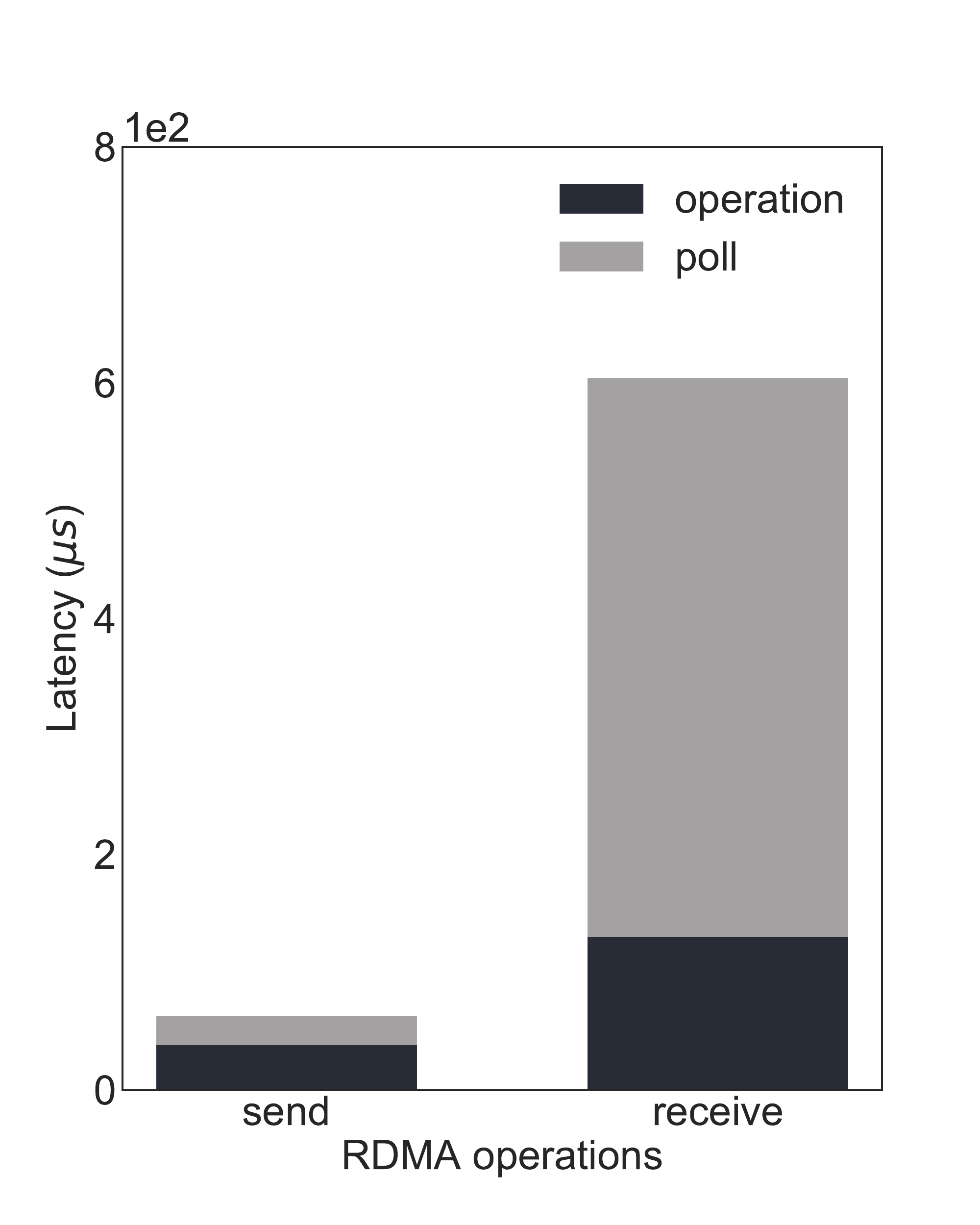}
	}
	\subfigure[Read/Write (2 MACHs)]{
		\label{fig:rw:dif}
		\centering
		\includegraphics[width=0.23\textwidth]{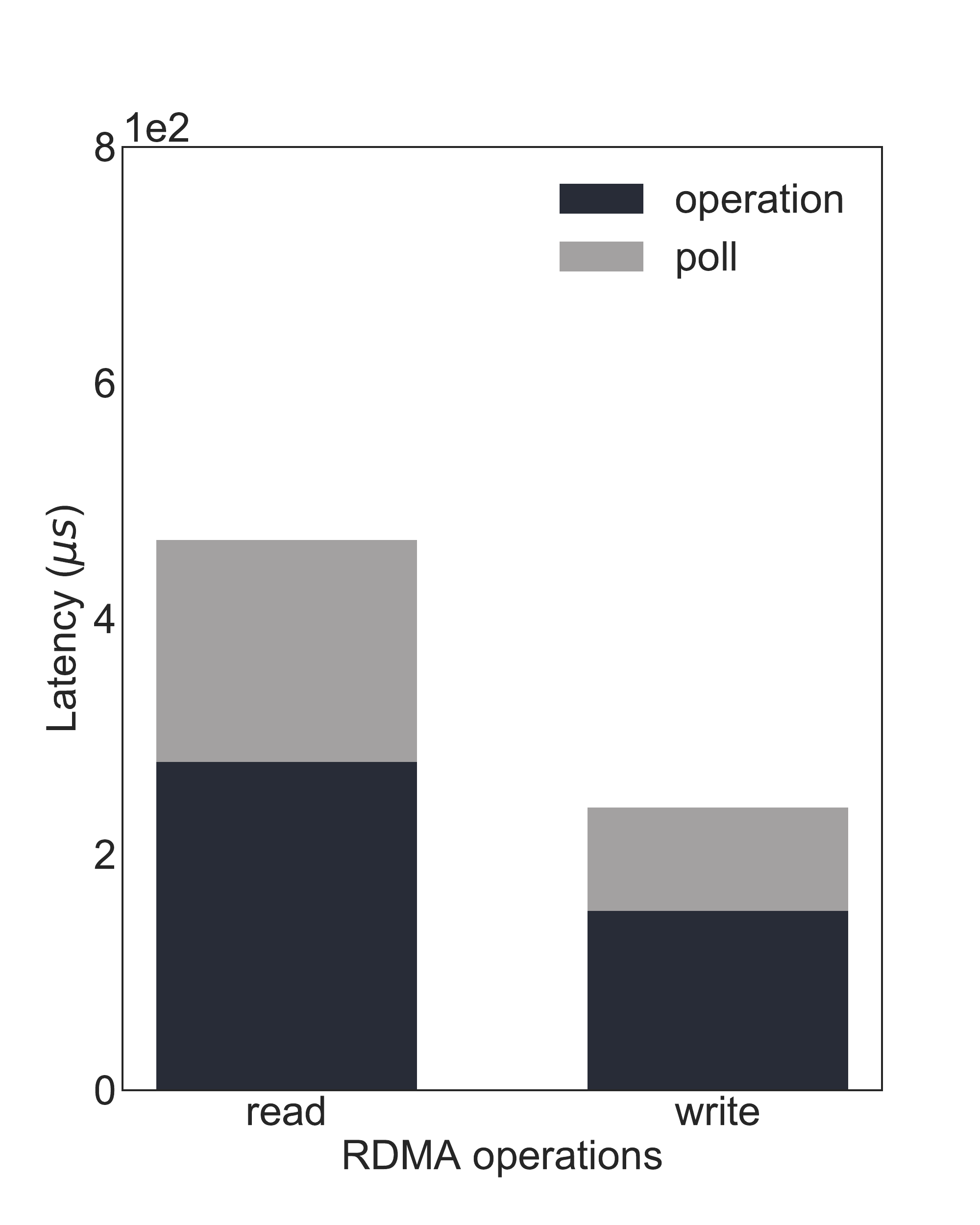}
	}
	\caption{Latency for basic RDMA operations}
	\label{fig:rdma:operation}
\end{figure}

\begin{figure*}[!ht]
	\centering
	\subfigure[Read and write.]{
		\label{fig:read:write:block}
		\centering
		\includegraphics[width=0.45\textwidth]{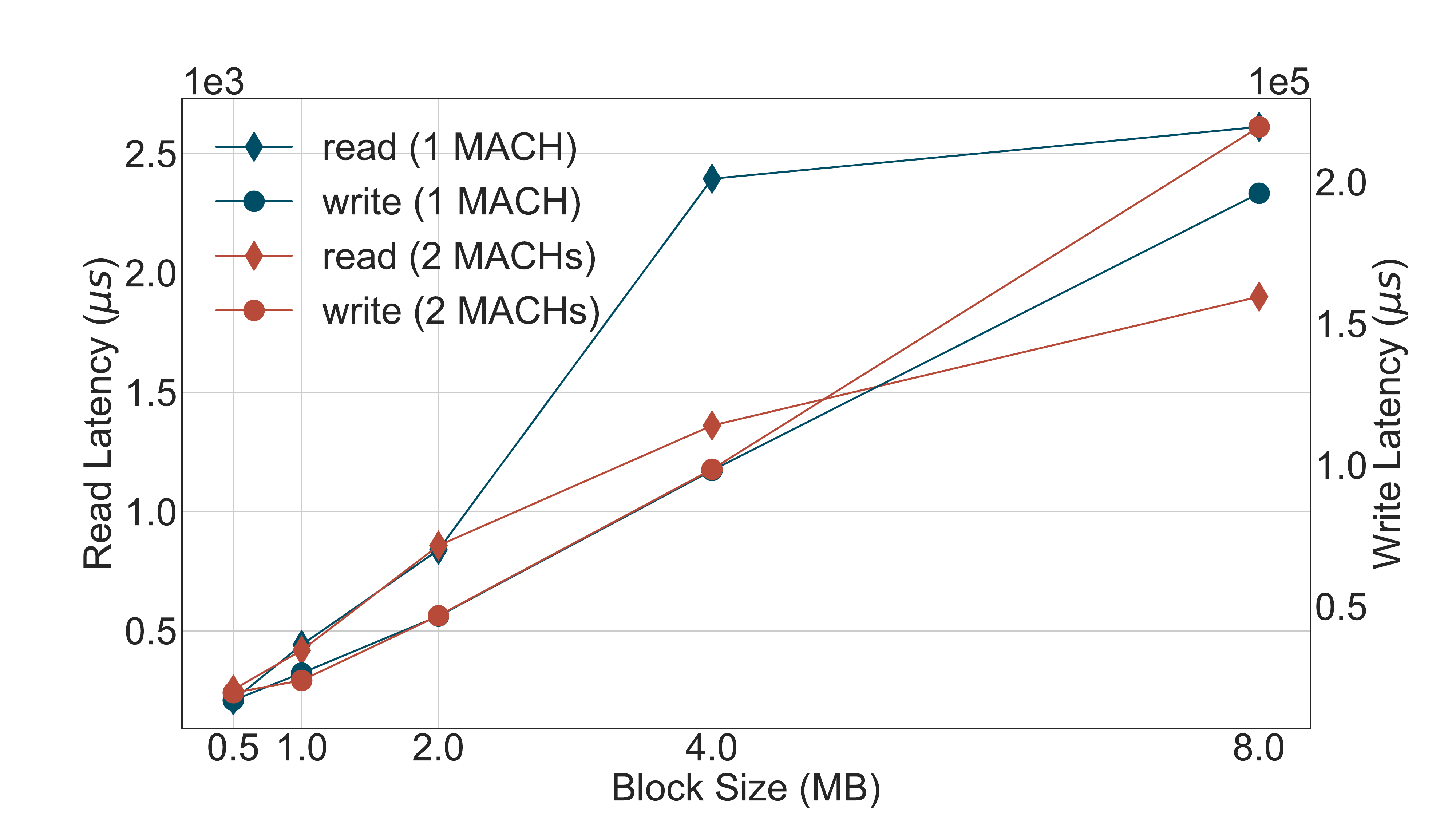}
	}
	\subfigure[Sign and verify.]{
		\label{fig:sign:verify:block}
		\centering
		\includegraphics[width=0.45\textwidth]{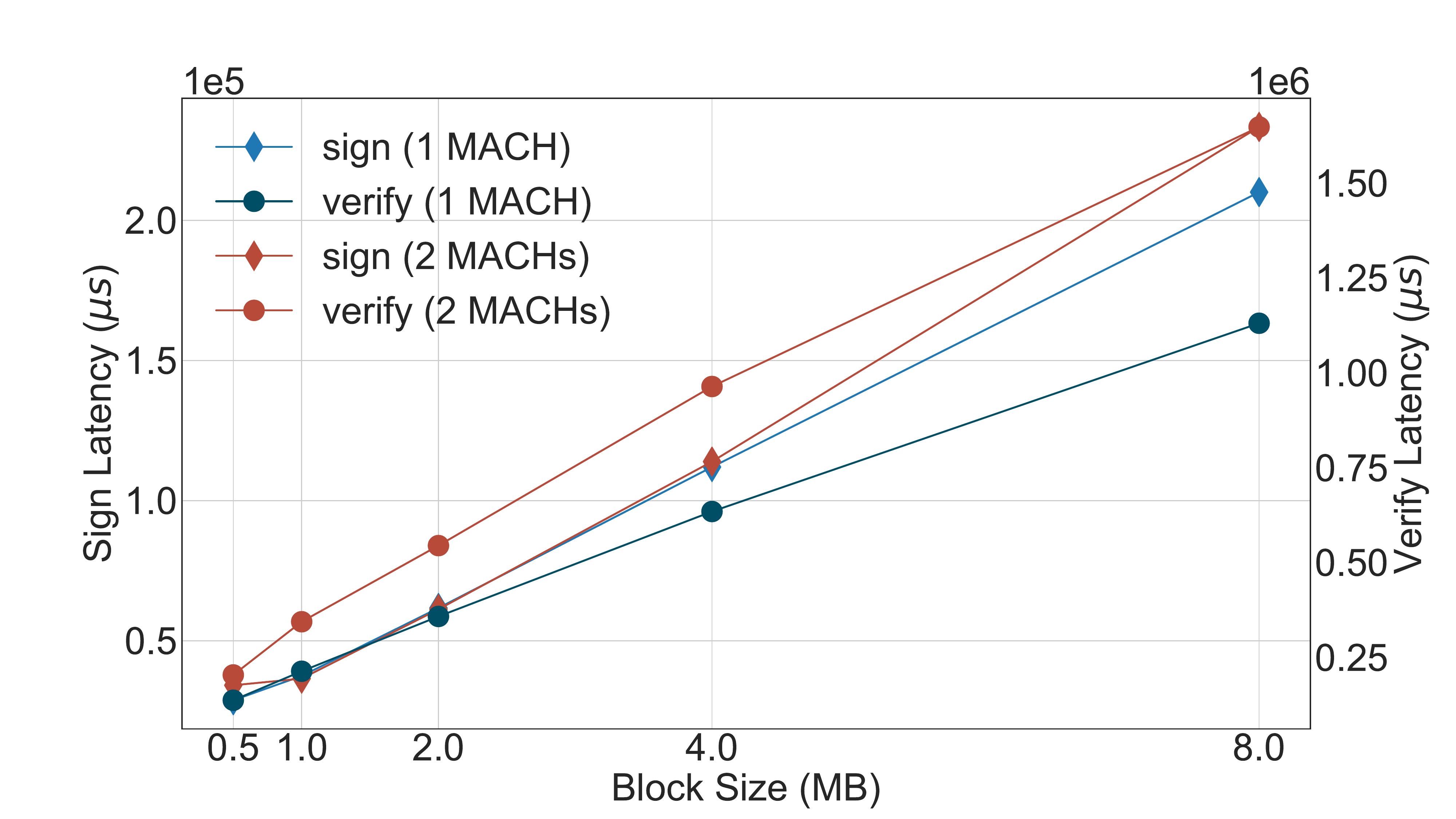}
	}
	\caption{Latency for Read/Write and Sign/Verify operations on a block with a varying block size.}
\end{figure*}

\begin{figure}[!ht]
	\centering
	\subfigure[Latency vs. network size.]{
		\label{fig:consensus:size}
		\centering
		\includegraphics[width=0.23\textwidth]{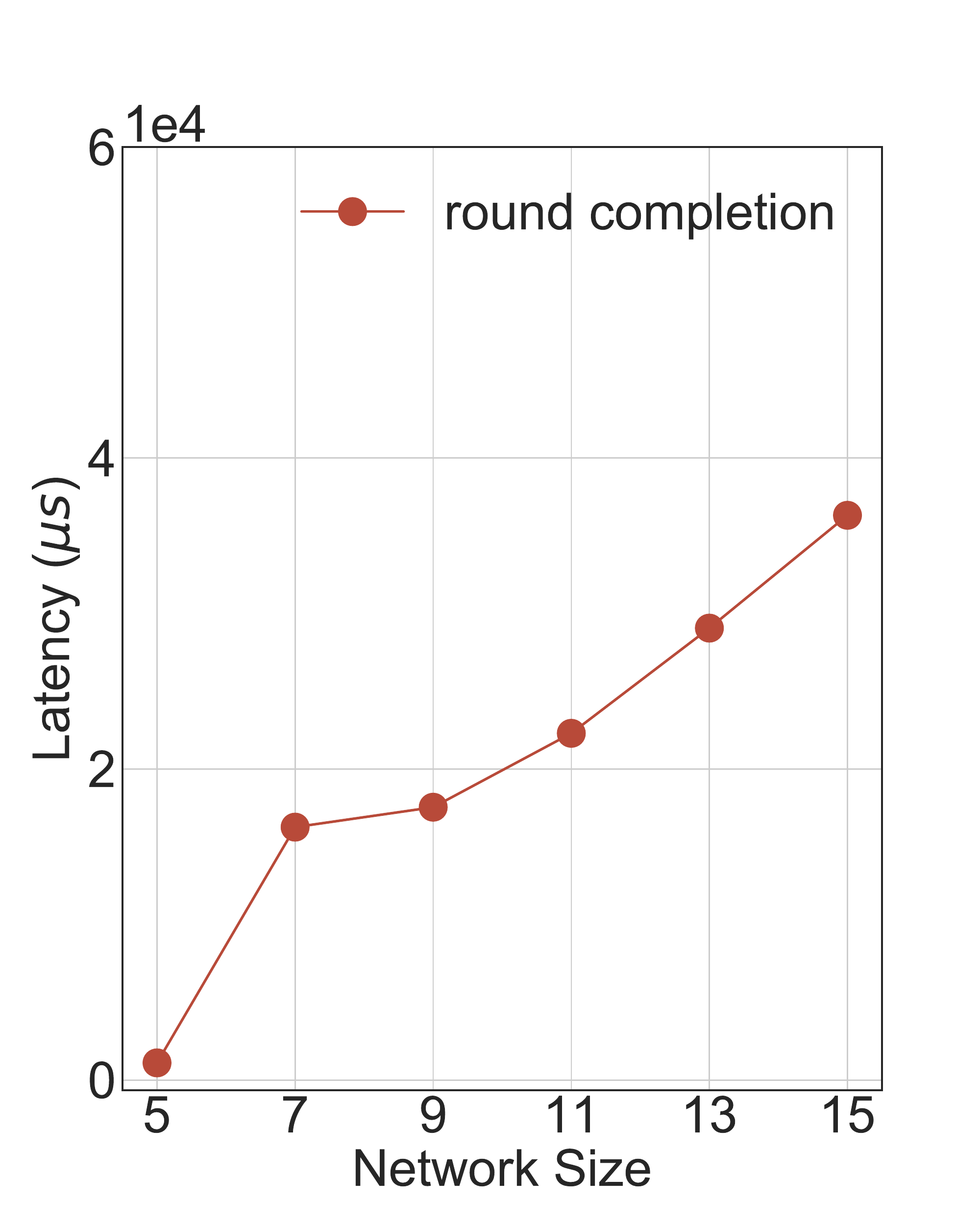}
	}
	\subfigure[Latency vs. Byzantine ratio.]{
		\label{fig:consensus:byzantine}
		\centering
		\includegraphics[width=0.23\textwidth]{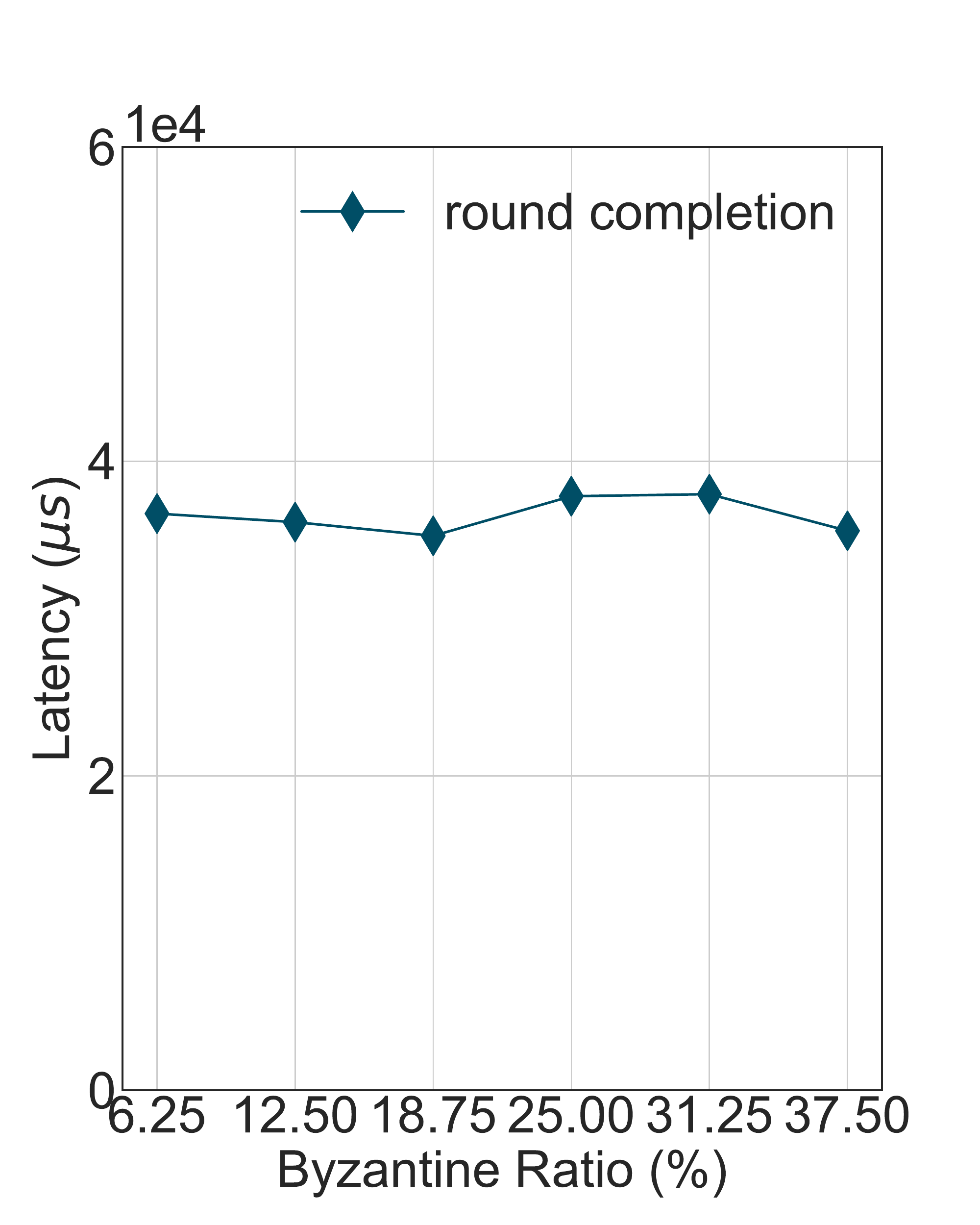}
	}
	\caption{Latency for the SM-based consensus algorithm.}
\end{figure}

\begin{figure*}[!htbp]
	\centering
	\subfigure[Latency and throughput vs. network size.]{
		\label{fig:blockchain:network:size}
		\centering
		\includegraphics[width=0.45\textwidth]{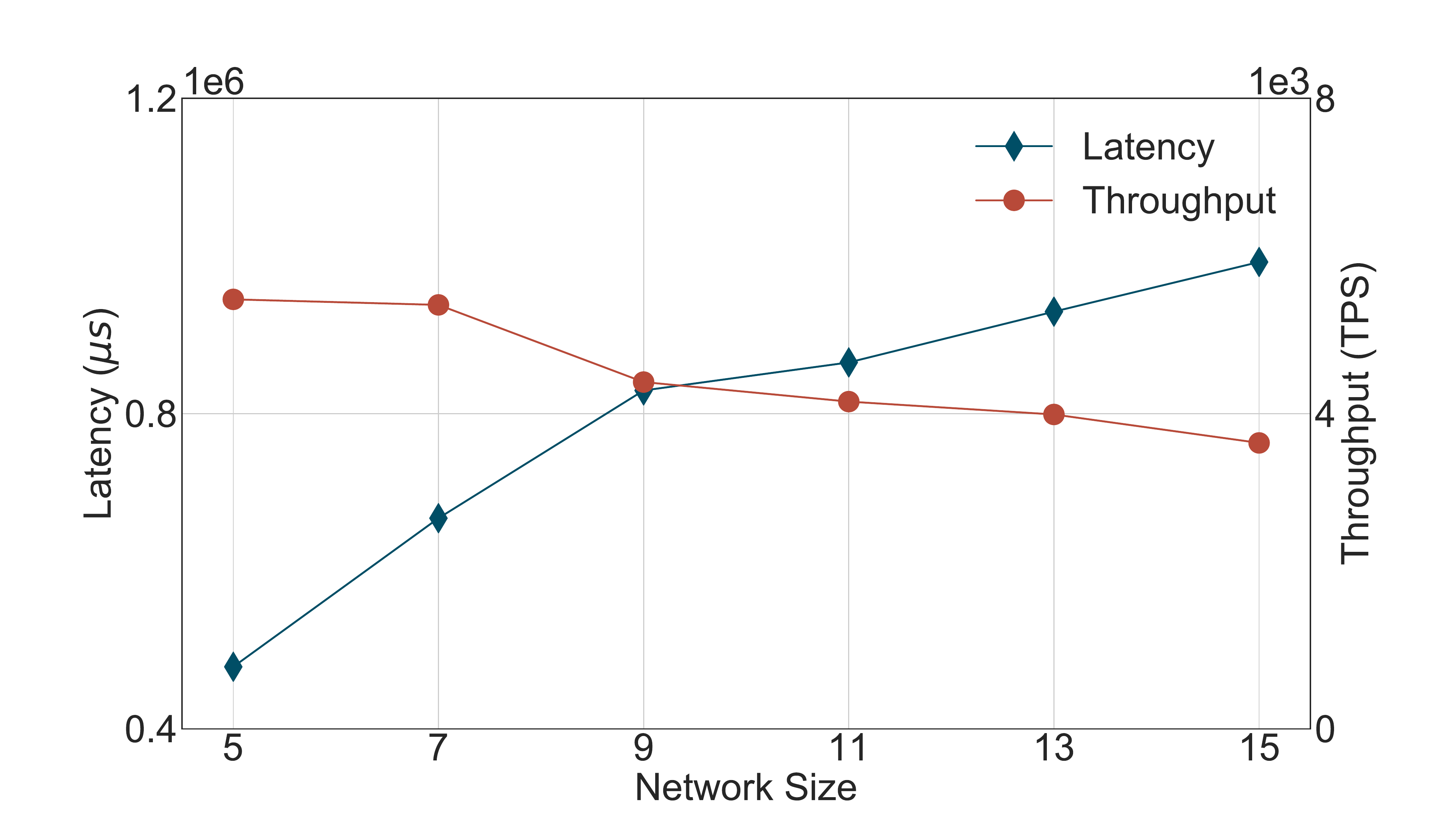}
	}
	\subfigure[Latency and throughput vs. block size.]{
		\label{fig:blockchain:block:size}
		\centering
		\includegraphics[width=0.45\textwidth]{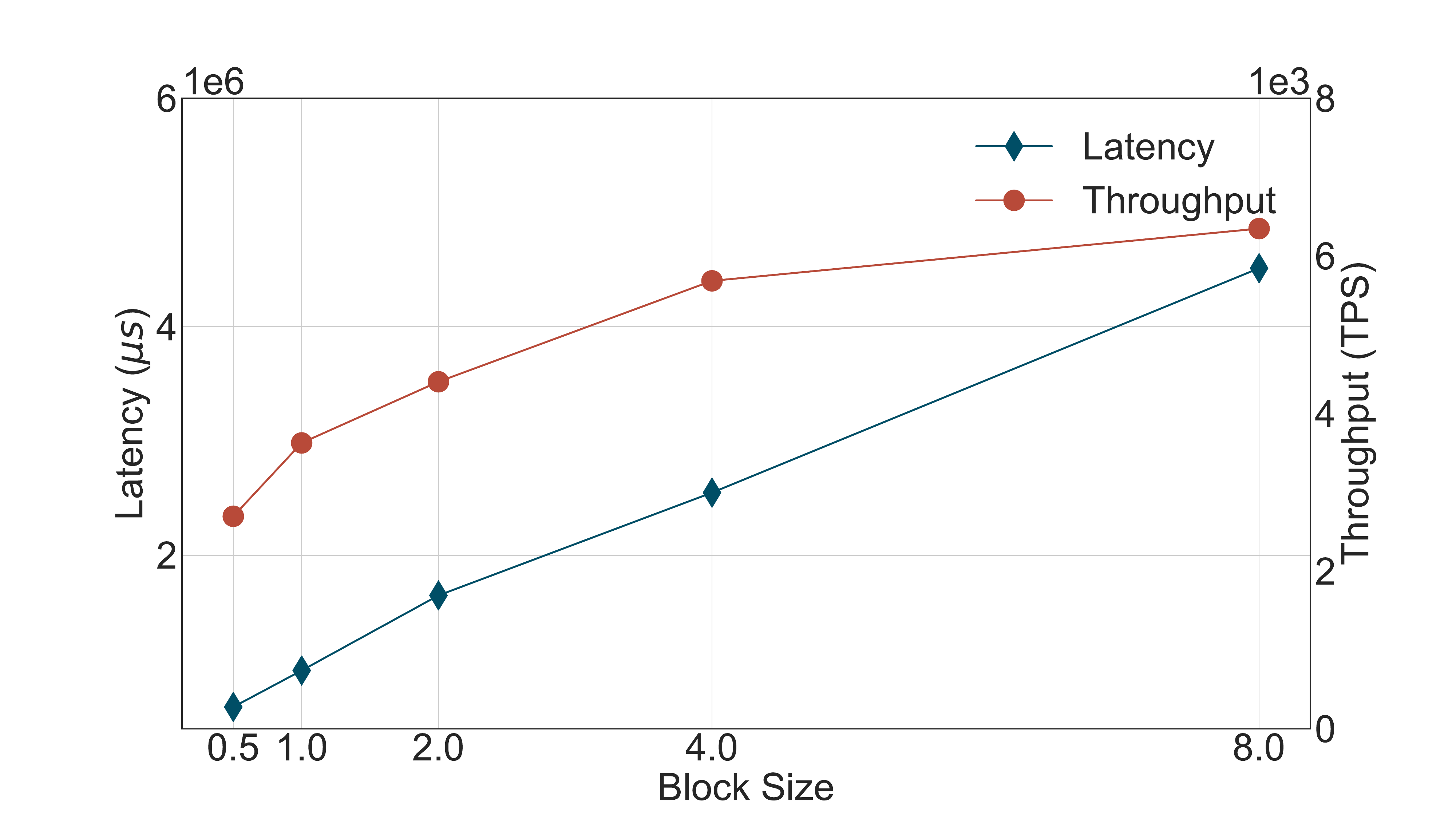}
	}
	\caption{Latency and throughput of CloudChain.}
\end{figure*}

\textbf{Commit Messages (Read/Write and Send/Receive).} We explore the performance of the send/receive and read/write operations on commit messages. Note that we only consider remote read/write and send/write messages as the local ones (within one VM) take negligible time\footnote{Even though CloudChain prohibits any node from writing on its peers' memory regions,  we still show the performance of remote writes for the completeness of our evaluations.}.  As shown in Fig.~\ref{fig:rdma:operation}, the legend ``poll'' is the time consumed by issuing and processing a request while ``operation'' is the pure transmission time (excluding the poll time). 
One can see that ``poll'' takes a quite large amount of time when receiving a message because all received messages are placed in a posted buffer and data is returned to the user though a work completion request.  A read operation takes a longer time than a write, which is anti-intuitive but is consistent with the result shown in \cite{DBLP:conf/hpcc/MacArthurR12}, an interesting paper that reports extensive experimental studies on RDMA performance.  Unfortunately, \cite{DBLP:conf/hpcc/MacArthurR12} does not reveal why  reads outperform writes,  but demonstrates that performance differences can happen when the message size is small and the network is not busy. 
On the other hand, we observe that operations executed by nodes on distinct servers cost more time than on a single physical server due to the time delay of remote memory access. 
We also notice that a pair of read/write operations have lower latency compared to a pair of send/receive operations. This explains why in our SMCA we only use read/write operations. 

\textbf{Block (Read/Write and Sign/Verify).} Then we test the latency of read/write and sign/verify operations on a block with varying size.  We observe that both read and write latency increases linearly with block size as shown in Fig.~\ref{fig:read:write:block}. A block of size 0.5 MB has the read latency of 0.210 ms for 1 MACH and 0.254 ms for 2 MACHs while the write latency is 16.951 ms for 1 MACH and 19.601 ms for 2 MACHs. For a block of size 8 MB, the read latency is 2.612 ms for 1 MACH and 1.901 ms for 2 MACHs while the write latency is 196.018 ms for 1 MACH and 219.450 ms for 2 MACHs. 

Fig.~\ref{fig:sign:verify:block} presents the latency of signing and verifying a block. The process of a sign operation includes signing a block and writing the block into a memory location while a verify operation consists of reading a block from a peer's memory and verifying the signature as well as the block header of that block.  Fig.~\ref{fig:read:write:block} indicates that the sign/verify time overhead increases linearly with block size.  When a block is of normal size (0.5-2 MB), its read and write operations are quite efficient as they can be completed within 0.1 second.   For example. the commonly adopted block size of 2 MB leads to the sign latency of 61.703 ms for 1 MACH and 357.824 ms for 2 MACHs while  verify latency is 61.246 ms  for 1 MACH  and 544.610 ms for 2 MACHs.  
A closer inspection of the figure reveals that nodes on two severs display larger latency compared to those on a single machine. This is because operations on two MACHs cost more time on block transmissions, which is consistent with the results in Fig.~\ref{fig:rdma:operation}.

\subsection{Performance of SMCA and CloudChain}
\label{sec:evaluation:CloudChain}
\textbf{Latency of SMCA.} We first investigate the latency of a one-round execution of SMCA with varying network size and Byzantine ratio.  We test $N=5:2:15$ nodes with $(N+1)/2$ of them located on server A and $(N-1)/2$ on server B.  Fig.~\ref{fig:consensus:size} indicates that the consensus latency increases with the network size, owing to the limited bandwidth.  Fig.~\ref{fig:consensus:byzantine} reports the results for the case when we equally distribute 16 nodes on two servers and set an honest leader on server A while the other 15 nodes can be either Byzantine or honest  according to the Byzantine ratio, which is defined to be the number of Byzantine nodes over the network size.  With the Byzantine ratio varying from 6.25\% to 37.50\%, the latency of consensus is nearly a constant.  This implies that when the leader is not Byzantine, the consensus latency is not obviously impacted by the number of Byzantine nodes.

\textbf{Throughput and Latency of CloudChain.} To determine the performance of CloudChain, we first explore how its latency and throughput are impacted by the network size.  Fig.~\ref{fig:blockchain:network:size} shows our experimental results, scaling the network size from 5 to 15 with $(N+1)/2$ on server A and $(N-1)/2$ on server B.  On one hand, the latency increases with the network size but the growth rate slows down when the number of nodes is above a certain value because the network latency itself is not linearly increased with size. With 15 nodes, the latency of confirming a 2 MB block is about 992.277 ms.  On the other hand, the throughput decreases from 5379 TPS to 3628 TPS, owing to the latency increase. This result can be explained by the fact that a larger network costs more time of scanning the votes.  Another possible reason lies in that we only deploy two servers thus increasing the number of nodes results in heavier contentions on the RDMA-enbaled channel.

Finally, we need to figure out how the performance is impacted by block size. Fig.~\ref{fig:blockchain:block:size} demonstrates that the latency increases with the increasing block size and reaches about 4512.822 ms for an 8 MB block. For a regular block of size 2 MB, the latency is about 1649.237 ms, which implies that CloudChain can commit about 2183 MB transactions per hour\footnote{Bitcoin has a rough throughput of 6 MB per hour.}. The throughput of CloudChain also increases with block size -- a regular block of size 2 MB corresponds to 4404 TPS and an 8 MB block contributes to 6347 TPS. Thus CloudChain can achieve higher throughput by increasing the block size without worrying about vulnerabilities caused by the block transmission delay.

\section{Conclusion and Future Research}
\label{sec:conclusion}
CloudChain is the first system that deeply integrates blockchain with the shared memory and RDMA technologies in cloud to provide high performance while preserving decentralization and protecting against Byzantine adversaries.  
CloudChain takes a modularized design with three layers: network layer, consensus layer, and blockchain layer, which together support three types of blockchain clients to accommodate heterogeneous device capacities.  Besides, CloudChain leverages RDMA to realize a shared-memory communication model, based on which we develop SMCA to achieve an agreement on the ordering of blocks against Byzantine behaviors.  We theoretically prove that CloudChain possesses persistence and liveness, the two most critical security properties of blockchain systems, indicating that CloudChain is securely strong enough to counter a variety of attacks.  Finally, we implement a CloudChain prototype and carry out a thorough experimental study on its performance, where the results validate our design effectiveness. 

In future, we would improve the performance of SMCA by introducing more advanced shared-memory techniques such as mutual exclusion.  Besides, it is promising to design cloud blockchains based on the M\&M model \cite{DBLP:conf/podc/AguileraBCGPT18}, which enjoys the nice features of both shared-memory and message-passing models.

\section*{Acknowledgment}
This study was partially supported by the National Key R\&D
Program of China under grant  2019YFB2102600, the National Natural Science Foundation of China under Grants 61771289 and  61832012, and the Blockchain Core Technology Strategic Research Program of Ministry of Education of China under grant  2020KJ010301.

\bibliographystyle{IEEEtran}
\bibliography{ref}

\end{document}